\documentclass[a4paper]{article}

\usepackage[english]{babel}
\usepackage[utf8]{inputenc}
\usepackage{amsmath}
\usepackage{amssymb}
\usepackage{amsthm}
\usepackage{graphicx}
\usepackage{hyperref}
\usepackage[colorinlistoftodos]{todonotes}

\newcommand{\undertilde}[1]{\underset{\widetilde{}}{#1}}

\newtheorem{prop}{Proposition}

\title{Correlation between graphs with an application to brain networks analysis}

\author{Andr\'e Fujita, Daniel Yasumasa Takahashi,\\ Joana Bisol Balardin, and Jo\~ao Ricardo Sato}

\begin{document}
\maketitle

\begin{abstract}
The global functional brain network (graph) is more suitable for characterizing brain states than local analysis of the connectivity of brain regions. Therefore, graph-theoretic approaches are the natural methods to study the brain. However, conventional graph theoretical analyses are limited due to the lack of formal statistical methods for estimation and inference for random graphs. For example, the concept of correlation between two vectors of graphs is yet not defined. The aim of this article to introduce a notion of correlation between graphs. In order to develop a framework to infer correlation between graphs, we assume that they are generated by mathematical models and that the parameters of the models are our random variables. Then, we define that two vectors of graphs are independent whether their parameters are independent. The problem is that, in real world, the model is rarely known, and consequently, the parameters cannot be estimated. By analyzing the graph spectrum, we showed that the spectral radius is highly associated with the parameters of the graph model. Based on it, we constructed a framework for correlation inference between graphs and illustrate our approach in a functional magnetic resonance imaging data composed of 814 subjects comprising 529 controls and 285 individuals diagnosed with autism spectrum disorder (ASD). Results show that correlations between default-mode and control, default-mode and somatomotor, and default-mode and visual sub-networks are higher ($p<0.05$) in ASD than in controls.
\end{abstract}

\section{Introduction}
\label{sec:intro}

Cutting-edge brain mapping techniques such as functional magnetic resonance imaging (fMRI) generate huge amounts of datasets that allows the construction of whole brain functional networks. Attempts to analyze and quantitatively characterize the structural properties of these networks are based on techniques of an emergence new field, namely complex network analysis (\cite{Boccaletti}; \cite{Newman}; \cite{Strogatz}).

Complex network analysis originated from mathematics, more specifically in graph theory, and aims to characterize the whole brain networks with a few number of measures. In this approach, a brain network is represented by a graph, in which its vertices represent the brain regions of interest (ROI), and edges represent the functional associations between ROIs (e.g. functional connectivity). Various graph-theoretic metrics can be used to investigate the mechanisms underlying the functional brain networks. Some examples are measures of functional integration, network motifs, centrality, and network resilience (\cite{Rubinov}). The analysis of the structural properties of the graphs allow us to visualize and understand the overall connectivity pattern of ROIs and also to quantitatively characterize its organization. These approaches became more popular over the last decade after it has provided an essential framework to elucidate the relationship between brain structure and function, and also to have proven by an increasing number of studies to give insights regarding the potential mechanisms involved in aging (\cite{Perry2015}), sex differences (\cite{Ingalhalikar2014}), various brain disorders (\cite{Stam2014}), and structural reconfiguration of the brain in response to external task modulation (\cite{Sherwin2015}).

Although applications of methods developed in graph theory have been successful in the analysis of brain networks as aforementioned, there is still a gap between these graph-based computational approaches and Statistics. For example, to the best of our knowledge, the concept of correlation between graphs is unknown. The concept of correlation between graphs may aid the understanding of how brain sub-networks interact and also to identify differences in those interactions between controls and patients (subjects diagnosed with a disorder) that may be useful for the development of novel procedures for diagnosis and prognosis.

Graphs are difficult to be manipulated from a statistical viewpoint because they are not numbers, but objects composed of one set of vertices and one set of edges. By observing the graphs depicted in Figure \ref{figure:illustration}, it is very difficult to identify correlation between them by only analyzing their structures. Thus, to construct a framework to infer correlation, one natural idea would be to imagine that a graph is generated by a mathematical model with a set of parameters. The parameters are the random variables. Intuitively, two vectors of graphs are correlated whether the parameters (random variables) of the graph model are correlated (Figure \ref{figure:illustration}). However, given two vectors of graphs, the model that generates them is rarely known, and consequently, the parameters cannot be estimated. Thus it is necessary to identify a feature of the graph that is highly associated with the parameters of the graph. In order to identify the feature that contains the information of the parameter, we investigated the spectral properties of random graphs (set of eigenvalues of its adjacency matrix). It is known that some structural properties, such as the number of walks, diameter, and cliques can be described by the spectrum of the graph (\cite{Mieghem}). Here, we propose to estimate the correlation between graphs by using the spectral radius (largest eigenvalue) of the graphs. Our results show that the spectral radius is highly associated with the parameters that generate the graph, and thus, it can be a good feature to calculate correlation between two graphs.

\begin{figure}
\centering
\includegraphics[width=4in]{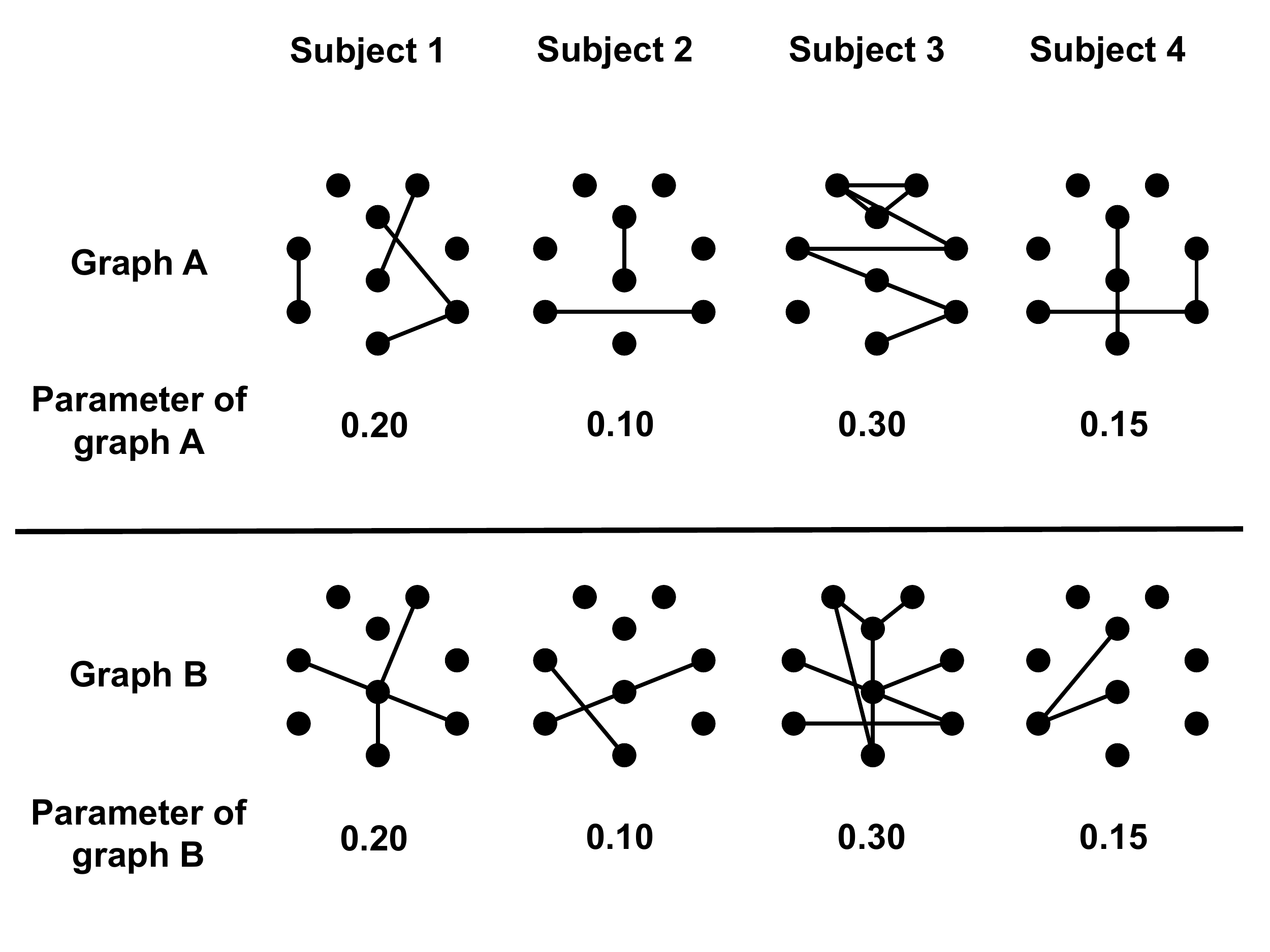}
\caption{Two vectors of perfectly correlated random graphs {\bf A} and {\bf B}, each one of size four. Graphs {\bf A} and {\bf B} are caricatural representations of two distinct brain sub-networks (e.g. somatomotor and default-mode) of four subjects. The identification of correlation by directly analyzing the structure of graphs {\bf A} and {\bf B} is very difficult. Notice that although they are generated by the same model (in this illustration, the graphs were generated by an Erd\"os-R\'enyi random graph model) and parameters, they are structurally different. Thus, one solution to identify correlation between graphs consists in identifying correlation between the parameters of the random graph models. \label{figure:illustration}}
\end{figure}



We illustrate the usefulness of our method by analyzing a large fMRI dataset (ABIDE - The Autism Brain Imaging Data Exchange - Consortium website - \url{http://fcon_1000.projects.nitrc.org/indi/abide/}) composed of 814 participants comprising 529 controls and 285 individuals with ASD.



\section{Description of the method}
\label{sec:meth}

\subsection{Graph}
A graph is a pair of sets $G=(V,E)$, where $V$ is a set of $n$ vertices ($v_{1}, v_{2}, \ldots, v_{n}$) and $E$ is a set of $m$ edges that connect two vertices of $V$.

Any undirected graph $G$ with $n$ vertices can be represented by its adjacency matrix ${\bf A}^G$ with $n \times n$ elements ${\bf A}_{ij}^G$ ($i,j=1, \ldots, n$), whose value is ${\bf A}_{ij}^G = {\bf A}_{ji}^G = 1$ if vertices $v_i$ and $v_j$ are connected, and 0 otherwise. The spectrum of graph $G$ is the set of eigenvalues of its adjacency matrix ${\bf A}^G$. Thus, an undirected graph with $n$ vertices has $n$ real eigenvalues $\lambda_1 \ge \lambda_2 \ge \ldots \ge \lambda_n$.


\subsection{Correlation between graphs}
\label{sec:correlation}

Two random variables are statistically independent whether knowledge about one of them does not aid in the prediction of the other. On the other hand, if they are not independent, then the values of one of the variables can be predicted by information provided about the other.

In this study, we consider that the parameters of the graph models are random variables. Thus, we assume we have $k$ independent graphs randomly generated by the same graph model, but each one with a distinct set of parameters (sampled from a probabilistic distribution).

Let $k$ and $\Theta$ be the number of graphs and their parameters, respectively. Then, let $\undertilde{\theta}^1 = \{ \theta_{1}^{1}, \theta_{2}^{1}, \ldots, \theta_{k}^{1} \}$ and $\undertilde{\theta}^2 = \{ \theta_{1}^{2}, \theta_{2}^{2}, \ldots, \theta_{k}^{2} \}$ be two samples of random variables $\Theta^1$ and $\Theta^2$, respectively, ${\undertilde{G}}^1(\undertilde{\theta}^1) = \{ G_{1}^{1}(\theta_{1}^{1}), G_{2}^{1}(\theta_{2}^{1}), \ldots, G_{k}^{1}(\theta_{k}^{1})\}$ and ${\undertilde{G}}^2(\undertilde{\theta}^2) = \{ G_{1}^{2}(\theta_{1}^{2}), G_{2}^{2}(\theta_{2}^{2}), \ldots, G_{k}^{2}(\theta_{k}^{2})\}$ be two samples of random graphs constructed by using $\undertilde{\theta}^1$ and $\undertilde{\theta}^2$, respectively. To illustrate this concept, suppose that ${\undertilde{G}}^1$ and ${\undertilde{G}}^2$ are two vectors of Erd\"os-R\'enyi random graphs (\cite{Erdos}). An Erd\"os-R\'enyi random graph (\cite{Erdos}) has $n$ labeled vertices in which each pair of vertices is connected by an edge with a given probability $p$. In this case, the probability $p$ is the parameter of graph $G$. Thus, the two vectors of Erd\"os-R\'enyi random graphs can be described as ${\undertilde{G}}^1(\undertilde{p}^1)=\{ G_{1}^{1}(p_{1}^{1}), G_{2}^{1}(p_{2}^{1}), \ldots, G_{k}^{1}(p_{k}^{1})\}$ and ${\undertilde{G}}^2(\undertilde{p}^2) = \{ G_{1}^{2}(p_{1}^{2}), G_{2}^{2}(p_{2}^{2}), \ldots, G_{k}^{2}(p_{k}^{2})\}$.

We say that random graphs ${\undertilde{G}}^1(\undertilde{\theta}^1)$ and ${\undertilde{G}}^2(\undertilde{\theta}^2)$ are independent if the vectors of parameters ${\undertilde{\theta}}^1$ and ${\undertilde{\theta}}^2$ are independent. In our example for an Erd\"os-R\'enyi random graph, we say that graphs ${\undertilde{G}}^1(\undertilde{p}^1)$ and ${\undertilde{G}}^2(\undertilde{p}^2)$ are correlated if the vectors of probabilities ${\undertilde{p}}^1=\{p_{1}^{1}, \ldots, p_{k}^{1}\}$ and ${\undertilde{p}}^2=\{p_{1}^{2}, \ldots, p_{k}^{2}\}$ are correlated.

Formally, two random variables $\Theta^{1}$ and $\Theta^{2}$ with probability density functions $f(\Theta^1)$ and $f(\Theta^2)$ are independent if and only if the combined random variable $(\Theta^1, \Theta^2)$ has a joint probability density function $f(\Theta^1,\Theta^2)=f(\Theta^1) \times f(\Theta^2)$. We say that two random variables $\Theta^1$ and $\Theta^2$ are dependent if they are not independent.

The test of independence between $G^1$ and $G^2$ is described as a hypothesis test as follows:\\
$\text{H}_0$: $\Theta^1$ and $\Theta^2$ are independent (null hypothesis)\\
$\text{H}_1$: $\Theta^1$ and $\Theta^2$ are not independent (alternative hypothesis)

One simple manner to identify correlation between $\Theta^1$ and $\Theta^2$ consists in, if the graph model is known, to estimate the parameters of the graphs, and then test the probabilistic dependence between them. However, the graph model is rarely known for real world graphs. Thus, the problem consists in detecting dependence only from the observation of random graphs (and not the parameters). In other words, it is necessary to identify a feature of the graph that is highly associated with the parameters of the graph.

From spectral graph theory, the largest eigenvalue ($\lambda_1$) of a graph $G$ is known as its spectral radius or index (for simplicity, we will denote the largest eigenvalue $\lambda_1$ just as $\lambda$). For several random graphs, it is known that the spectral radius is a function of the parameters of the graph. For example, for the Erd\"os-R\'enyi random graph, let $n$ and $p$ be the number of vertices and the probability that two vertices are connected by an edge, respectively. Then, the spectral radius of an Erd\"os-R\'enyi random graph is $np$. Thus, we propose to use the spectral radius to identify correlation between graphs.

Let ${\undertilde{G}}^1 = \{ G_{1}^{1}, G_{2}^{1}, \ldots, G_{k}^{1},\}$ and ${\undertilde{G}}^2 = \{ G_{1}^{2}, G_{2}^{2}, \ldots, G_{k}^{2},\}$ be two samples of random graphs and $\undertilde{\lambda}^1 = \{ \lambda_{1}^{1}, \lambda_{2}^{1}, \ldots, \lambda_{k}^{1} \}$ and $\undertilde{\lambda}^2 = \{ \lambda_{1}^{2}, \lambda_{2}^{2}, \ldots, \lambda_{k}^{2} \}$ be the spectral radii associated with ${\undertilde{G}}^1$ and ${\undertilde{G}}^2$, respectively. Thus, to identify correlation between graphs, one may test the independence between $\undertilde{\lambda}^1$ and $\undertilde{\lambda}^2$.



\subsubsection{Identification of the correlation between graphs} \label{section:spearman_correlation}
Once defined the feature to be used to identify the correlation between graphs, then it is necessary to estimate the correlation itself. We propose the use of the Spearman's rank correlation ($\rho$) because: (i) its implementation is simple; (ii) it is robust to outliers; and (iii) it does not require assumptions of linearity in the relationship between variables (it can identify monotonic nonlinear associations), nor the variables should be measured at interval scales, as it can be used for ordinal variables (\cite{spearman1904general}).

Let $\hat{\rho}$ be the sample Spearman's rank correlation coefficient. To estimate $\hat{\rho}$, first convert the raw values of $\lambda^1_i$ and $\lambda^2_i$ ($i=1, \ldots, k$) to ranks, and calculate the differences $d_i$ between the ranks of $\lambda^1_i$ and $\lambda^2_i$. Then, calculate the Spearman's rank correlation coefficient ($\hat{\rho}$) as:

\begin{math}
\hat{\rho} = 1-\frac{6\sum_{i=1}^{k}d_i^2}{k(k^2-1)}
\end{math}

\noindent where $\lim_{k \rightarrow \infty}\mathbb{E}[\hat{\rho}(\undertilde{\lambda}^1, \undertilde{\lambda}^2)] := \rho(\lambda^1, \lambda^2) = 12 \mathbb{E}[F_1(\lambda^1)F_2(\lambda^2)] - 3$. 

Observe that if $\lambda^1$ and $\lambda^2$ are independent, then $\rho(\lambda^1, \lambda^2) = 0$. The main idea of this article is that in several cases $12 \mathbb{E}[F_1(\lambda^1)F_2(\lambda^2)] - 3 = 0$ if and only if $12 \mathbb{E}[F_1(\Theta^1)F_2(\Theta^2)] - 3 = 0$ therefore we can use $\hat{\rho}(\undertilde{\lambda}^1, \undertilde{\lambda}^2)$ to estimate $12 \mathbb{E}[F_1(\Theta^1)F_2(\Theta^2)] - 3$.

We can prove the validity of this idea in a simple case. Denote by $F$ the joint probability distribution for $(\Theta^1, \Theta^2)$ and the marginals for $\Theta^1$ and $\Theta^2$ by $F_1$ and $F_2$, respectively. 
\begin{prop}
Let $F$ be differentiable on both coordinates. Given i.i.d. copies $(\Theta^1_i, \Theta^2_i)_{i = 1,\ldots, k}$ of $(\Theta^1, \Theta^2)$, let $(G^1_i(\Theta^1_i))_{i = 1,\ldots, k}$ and $(G^2_i(\Theta^2_i))_{i = 1,\ldots, k}$ be independent ER random graphs of size $n$. Then, for any positive $\epsilon, \delta$ there exist an integer $k_0$ such that for all $k > k_0$ and $n > n_0(k)$ we have with probability larger than $1-\delta$ that
\begin{equation}
 \left|\hat{\rho}(\undertilde{\lambda}^1, \undertilde{\lambda}^2) - 12 \mathbb{E}[F_1(\Theta^1)F_2(\Theta^2)] +3\right| \leq \epsilon 
 \end{equation}
\end{prop}

\begin{proof}
Let $\mathbb{P}$ be the joint probability measure for the sequences $(\Theta^1_i)_{i \geq 1}$, $(\Theta^2_i)_{i \geq 1}$, $(G^1_i(\Theta^1_i))_{i \geq 1}$, and $(G^2_i(\Theta^2_i))_{i \geq 1}$. To prove the proposition, it is enough to show that for suitable choices of $k$ and $n$, we have
\begin{equation}\label{eq:boundeig}
\mathbb{P}\left(\left|\hat{\rho}(\undertilde{\lambda}^1, \undertilde{\lambda}^2) - \hat{\rho}(\undertilde{\Theta}^1, \undertilde{\Theta}^2) \right| > \epsilon/2\right)< \delta/2
\end{equation}
and
\begin{equation} \label{eq:boundrho}
 \mathbb{P}\left(\left|\hat{\rho}(\undertilde{\Theta}^1, \undertilde{\Theta}^2) - 12 \mathbb{E}[F_1(\Theta^1)F_2(\Theta^2)] +3\right| > \epsilon/2 \right) < \delta/2.
\end{equation} 

It is a classical result (see for example \cite{borkowf2002computing}) that $\hat{\rho}((\Theta^1_i)_{i = 1,\ldots, k}, (\Theta^2_i)_{i = 1,\ldots, k})$ converges in probability to  $12 \mathbb{E}[F_1(\Theta^1)F_2(\Theta^2)] -3$, therefore, for sufficiently large $k$, we have that \eqref{eq:boundrho} holds.

Now, it remains to prove that there is $n_0(k)$ such that for all $n > n_0(k)$ inequality \eqref{eq:boundeig} holds. Let $\hat{r}_i^1$ and $\hat{r}_i^2$ for $i =1, \ldots, k$ be the ranks of the spectral radii of the graphs $G^1_i(\Theta^1_i)$ and $G^2_i(\Theta^2_i)$, respectively. Also, let $r_i^1$ and $r_i^2$ for $i =1, \ldots, k$ be the ranks of the $\Theta^1_i$ and $\Theta^2_i$, respectively. 
From the definition of Spearman correlation, it is clear that if $\hat{r}_i^l = r_i^l$ for $i=1,\ldots, k$ and $l =1,2$, we have that 
\begin{equation*} 
\hat{\rho}(\undertilde{\lambda}^1, \undertilde{\lambda}^2) = \hat{\rho}(\undertilde{\Theta}^1, \undertilde{\Theta}^2).
\end{equation*}

Therefore,  we have that
\begin{align*}
&\mathbb{P}\left(\left|\hat{\rho}(\undertilde{\lambda}^1, \undertilde{\lambda}^2)  - \hat{\rho}(\undertilde{\Theta}^1, \undertilde{\Theta}^2) \right| > \epsilon/2\right)\\
& \leq \mathbb{P}\left(\hat{\rho}(\undertilde{\lambda}^1, \undertilde{\lambda}^2)  \neq \hat{\rho}(\undertilde{\Theta}^1, \undertilde{\Theta}^2)\right)\\
&\leq \mathbb{P} \left(\hat{r}^l_i \neq r^l_i \;\; \text{for some} \;\; i=1,\ldots, k\;\; \text{and}\;\; l =1,2 \right)\\
& \leq \sum_{l =1}^2\sum_{i=1}^k\mathbb{P}\left( \hat{r}^l_i \neq r^l_i \right)
\end{align*}

To prove \eqref{eq:boundeig},  using the above inequalities, it is enough to show for $l = 1,2$ and $i = 1, \ldots, k$ that 
\begin{equation*}
\mathbb{P}\left( \hat{r}^l_i \neq r^l_i \right) < \frac{\delta}{4k}.
\end{equation*}

Let $\gamma = \min\{| \Theta^l_i - \Theta^l_j |: i, j = 1, \ldots, k \;\; \text{and}\; l = 1,2\}$. We have that
\begin{equation*}
\mathbb{P}\left( \hat{r}^l_i \neq r^l_i \right)  \leq \mathbb{P}\left(\left|\frac{\lambda_i^l}{n} - \Theta^l_i \right| > \gamma/2 \right).
\end{equation*}
For an increasing sequence of ER random graph with parameter $p$ and their respective spectral radii $\lambda(n)$, it is well known that $\lambda(n)/n$ converges in probability to $p$ (\cite{ding2010spectral}), therefore, for large enough $n_0$ we have that for all $n > n_0$
\begin{equation*}
\mathbb{P}\left(\left|\frac{\lambda_i^l}{n} - \Theta^l_i \right| > \gamma/2 \right) < \frac{\delta}{4k}.
\end{equation*}
This concludes the proof of the proposition.
\end{proof}



Spearman's correlation coefficient ($\rho$) assumes values between -1 and 1, where $\rho$ is $+1$ in the case of a perfect monotonically increasing relationship and $-1$ in the case of a perfect monotonically decreasing relationship. In the case of imperfect monotonically dependence, $-1 < \rho < +1$, and in the case of monotonically independent random variables, $\rho=0$.

Thus, the hypothesis test to identify Spearman's correlation between two vectors of graphs can be defined as:

$\text{H}_0$: $\rho=0$ (null hypothesis)\\
$\text{H}_1$: $\rho\ne0$ (alternative hypothesis)

The Spearman's rank correlation coefficient under the null hypothesis can be asymptotically approximated by a Student's t-distribution with $k-2$ degrees of freedom as (\cite{spearman1904general}):

\begin{math}
t=\frac{\hat{\rho}\sqrt{k-2}}{\sqrt{1-\hat{\rho}^2}}
\end{math}

Now, suppose we have two conditions A and B and consequently, we have four graphs ${\undertilde{G}}^1_{\text{A}}$ and ${\undertilde{G}}^2_{\text{A}}$ in condition A, and ${\undertilde{G}}^1_{\text{B}}$ and ${\undertilde{G}}^2_{\text{B}}$ in condition B. We are interested in testing whether the correlation between graphs ${\undertilde{G}}^1_{\text{A}}$ and ${\undertilde{G}}^2_{\text{A}}$ are equal to the correlation between ${\undertilde{G}}^1_{\text{B}}$ and ${\undertilde{G}}^2_{\text{B}}$. In other words, we would like to test

$\text{H}_0: \rho_{\text{A}} = \rho_{\text{B}}$\\
$\text{H}_1: \rho_{\text{A}} \ne \rho_{\text{B}}$.

This test can be performed by using the procedure developed by \cite{Fisher}. Let $k_\text{A}$ and $k_\text{B}$ be the number of graphs in conditions A and B, respectively. First, transform each of the two correlation coefficients as $\hat{\rho}'_{\text{A}} = \frac{1}{2}\text{log}(\frac{1+\hat{\rho}_{\text{A}}}{1-\hat{\rho}_{\text{A}}})$ and $\hat{\rho}'_{\text{B}} = \frac{1}{2}\text{log}(\frac{1+\hat{\rho}_{\text{B}}}{1-\hat{\rho}_{\text{B}}} )$. Then, calculate the test statistic as $z=\frac{\hat{\rho}'_{\text{A}} - \hat{\rho}'_{\text{B}}}{\sqrt{ \frac{1}{k_\text{A} - 3} + \frac{1}{k_\text{B}-3}}}$. Finally, compute the $p$-value for the $z$-statistic.

In the present study, we used the \verb+R+ function \verb+cor.test+ with parameter \verb+method='spearman'+ (package \verb+stats+) to compute the Spearman's correlation test.


\section{Simulation studies} \label{section:simulations}

We carried out Monte Carlo simulations in five different random graphs to illustrate the performance of the proposed framework. Among several classes of random graphs, we describe the Erd\"os-R\'enyi random graph (\cite{Erdos}), random geometric graph (\cite{Penrose}), random regular graph (\cite{Meringer}), Barab\'asi-Albert random graph (\cite{Barabasi}), and Watts-Strogatz random graph (\cite{Watts}), due to their importance to model real world events or their well known theoretical results.


\subsection{Random graph models}

\subsubsection{Erd\"os-R\'enyi random graph} \label{section:random}

Erd\"os-R\'enyi random graphs (\cite{Erdos}) are one of the most studied random graphs. Erd\"os and R\'enyi defined a random graph as $n$ labeled vertices in which each pair of vertices $(v_i,v_j)$ is connected by an edge with a given probability $p$.

The spectral radius of an Erd\"os-R\'enyi random graph is $np$ (\cite{Furedi}).

The \verb+R+ function used to generate an Erd\"os-R\'enyi random graph is \verb+erdos.renyi.game+ (package \verb+igraph+). The \verb+igraph+ package can be downloaded from the \verb+R+ website (http://www.r-project.org).


\subsubsection{Random geometric graph}

A random geometric graph (RGG) is a spatial network. An undirected graph is constructed by randomly placing $n$ vertices in some topological space $\text{R}^d$ (e.g. a unit square - $d=2$) according to a specified probability distribution (e.g. uniform distribution) and connecting two vertices by an edge if their distance (according to some metric, e.g., Euclidian norm) is smaller than a certain neighborhood radius $r$. Hence, random geometric graphs have a spatial element absent in other random graphs.

The spectral radius of a random geometric graph converges almost surely to $r^{d}$ (\cite{Bordenave}).

The \verb+R+ function used to generate a geometric random graph is \verb+grg.game+ (package \verb+igraph+).



\subsubsection{Random regular graph}

A random regular graph is a graph where each vertex has the same number of adjacent vertices; i.e. every vertex has the same degree. A random regular graph with vertices of degree $deg$ is called a random $deg$-regular graph or random regular graph of degree $deg$ (\cite{Meringer}).

Random regular graphs of degree at most 2 are well known: a 0-regular graph consists of disconnected vertices; a 1-regular graph consists of disconnected edges; a 2-regular graph consists of disconnected cycles and infinite chains; a 3-regular graph is known as a cubic graph.

The spectral radius of a random $deg$-regular graph is $deg$ (\cite{Alon}).

The \verb+R+ function used to generate a regular random graph is \verb+k.regular.game+ (package \verb+igraph+).


\subsubsection{Barab\'asi-Albert random graph} \label{section:SF-graph}

Barab\'asi-Albert random graphs proposed by \cite{Barabasi} have a power-law degree distribution due to vertices preferential attachment (the more connected a vertex is, the more likely it is to receive new edges). \cite{Barabasi} proposed the following construction: start with a small number of ($n_0$) vertices and at every time-step, add a new vertex with $m_1$ ($m_1 \le n_0$) edges that connect the new vertex to $m_1$ different vertices already present in the system. When choosing the vertices to which the new vertex connects, assume that the probability that a new vertex will be connected to vertex $v_i$ is proportional to the degree of vertex $v_i$ and the scaling exponent $p_s$ ($P(v_i) \sim degree(v_i)^{p_s}$, where $degree(v_i)$ is the number of adjacent edges of vertex $v_i$ in the current time step) which indicates the order of the proportionality ($p_s=1$ linear; $p_s=2$ quadratic and so on).

Let $k_0$ be the smallest degree, the spectral radius of the Barab\'asi-Albert random graph is of the order of $k_{0}^{1/2}n^{1/2(p_s-1)}$ (\cite{Dorogovtsev}).

The \verb+R+ function used to generate a Barab\'asi-Albert random graph is \verb+barabasi.game+ (package \verb+igraph+).


\subsubsection{Watts-Strogatz random graph} \label{section:WS-graph}

Watts-Strogatz random graph (\cite{Watts}) is a random graph that interpolates between a regular lattice and an Erd\"os-R\'enyi random graph. This random graph present small-world properties (short average path lengths, i.e., most vertices are not neighbors of one another but can be reached from every other vertex by a small number of steps) and higher clustering coefficient (the number of triangles in the graph) than Erd\"os-R\'enyi random graphs.

The algorithm to construct a Watts-Strogatz random graph is as follows:

{\bf Input}: Let $n$, $nei$, and $p_w$ be the number of vertices, the number of neighbors (mean degree), and the rewiring probability, respectively.
\begin{enumerate}
\item construct a ring lattice with $n$ vertices, in which every vertex is connected to its first $nei$ neighbors ($\frac{nei}{2}$ on either side);
\item choose a vertex and the edge that connects it to its nearest neighbor in a clockwise sense. With probability $p_w$, reconnect this edge to a vertex chosen uniformly at random over the entire ring. This process is repeated by moving clockwise around the ring, considering each vertex in turn until one lap is completed. Next, the edges that connect vertices to their second-nearest neighbors clockwise are considered. As in the previous step, each edge is randomly rewired with probability $p_w$; continue this process, circulating around the ring and proceeding outward to more distant neighbors after each lap, until each edge in the original lattice has been considered once.
\end{enumerate}
{\bf Output}: the Watts-Strogatz random graph

To the best of our knowledge, the spectral radius of a Watts-Strogatz random graph is not analytically defined, but there are empirical evidences that it is a function of $p_w$ and $nei$ (\cite{Mieghem}).

The \verb+R+ function used to generate a Watts-Strogatz random graph is \verb+watts.strogatz.game+ (package \verb+igraph+).


\subsection{Simulation description} \label{simulation}
We designed three simulations to evaluate: (i) whether the Spearman's correlation between the spectral radii indeed retrieves the association between the parameters of the graph; (ii) the control of the rate of type I error and power of the method based on the spectral radius; and (iii) the performance of the Fisher's test on the spectral radius.

\subsubsection{Simulation 1}\label{simulation1}
In order to verify whether the Spearman's correlation coefficient between the spectral radii in fact retrieves the association between the parameters of the graph, we compared the correlation estimated directly from the parameters with the correlation obtained by analyzing the spectral radius. The design of the experiment is as follows: we set the graph model as the Erd\"os-R\'enyi random graph, the number of graphs as $k=50$ and the size of the graph as $n=100$. The parameter $p$ of the Erd\"os-R\'enyi random graphs were generated from a bivariate normal distribution with mean zero and covariance matrix $\Sigma=\begin{pmatrix} 1 & s\\ s & 1\end{pmatrix}$ with $s = -1.0, -0.9, \ldots,0, \ldots, 0.9, 1.0$.  Spearman's correlation coefficients are estimated by using the spectral radii. For each value of $s$, we repeated this procedure 30 times and compared the estimated correlation coefficient with the correlation ($s$) in fact used to generate the graphs.

\subsubsection{Simulation 2}\label{simulation2}
In order to evaluate the control of the rate of false positives under the null hypothesis and also its statistical power to identify correlation between two vectors of graphs, we constructed the following simulation study.

The parameters of the graphs are generated from a bivariate normal distribution with covariance matrix $\Sigma$ ($(\Theta^1, \Theta^2) \sim N(0, \Sigma)$). In order to evaluate the control of the rate of false positives under the null hypothesis (no correlation), we set $\Sigma=\begin{pmatrix} 1 & 0\\0 & 1\end{pmatrix}$. To evaluate the power of the test, we set $\Sigma=\begin{pmatrix} 1 & 0.5\\0.5 & 1\end{pmatrix}$. Notice that the vectors of parameters ($(\undertilde{\theta}^1, \undertilde{\theta}^2)$) must be linearly normalized in the interval $[0, 1]$. The set-up of the parameters of the graphs is as follows:

\begin{itemize}
\item Erd\"os-R\'enyi random graph: $p := \theta$
\item Random geometric graph: $r := \theta$ and $d=2$
\item Random regular graph: $deg$ := $\text{integer part of } 10 \times \theta$
\item Barab\'asi-Albert random graph: $k_0=3$ and $p_s := \text{integer part of } 10\times \theta$
\item Watts-Strogatz random graph: $nei=3$, and $p_w := \theta$
\end{itemize}

The number of graphs varied in $k = 20, 40, 60, 80, 100$. The size of the graphs was set to $n=50$.

The vectors of spectral radii ($\undertilde{\lambda}^1$ and $\undertilde{\lambda}^2$) were computed by using the adjacency matrix of each graph; and the Spearman's correlation test applied on both the parameters and the spectral radii.

This process was repeated 1,000 times for each number of graphs $k$ and pair of graph models. In order to evaluate and compare the power of the test between applying the correlation test on the parameters of the graph or on the spectral radii, we constructed receiver-operating characteristic (ROC) curves. The ROC curve is a bi-dimensional plot with the one minus the specificity (number of true negatives/(number of true negatives+number of false positives)) on the $x$-axis and the sensitivity (number of true positives/(number of true positives+number of false negatives)) on the $y$-axis. A curve above and further the diagonal means high power while a curve close to the diagonal means random decisions. In our case, the nominal $p$-value is on the $x$-axis and the proportion of rejected null hypothesis (the proportion of associations identified between two random variables), on the $y$-axis. ROC curves were plot (i) to verify the control of the rate of false positives; (ii) to evaluate the power of the test; and (iii) to compare the performance of the correlation estimated by using the spectral radius and the original parameter of the graph.

\subsubsection{Simulation 3}\label{simulation3}
In order to evaluate the performance of Fisher's test on both the control of the rate of false positives under the null hypothesis and also its statistical power between two conditions A and B, we constructed ROC curves for the following experimental set-up.

The graphs were constructed by using the Erd\"os-R\'enyi random graph model with the parameter $p$ generated by bivariate normal distributions with means zero and covariance matrices $\Sigma_{\text{A}}$ and $\Sigma_{\text{B}}$ for conditions A and B, respectively. In order to evaluate the control of the rate of false positives under the null hypothesis (same correlation between conditions A and B), we set $\Sigma_{\text{A}}=\Sigma_{\text{B}}=\begin{pmatrix} 1 & 0.5\\0.5 & 1\end{pmatrix}$. To evaluate the power of the test in identifying differences in correlation between conditions A and B, we set $\Sigma_{\text{A}}=\begin{pmatrix} 1 & 0.3\\0.3 & 1\end{pmatrix}$ and $\Sigma_{\text{B}}=\begin{pmatrix} 1 & 0.6\\0.6 & 1\end{pmatrix}$. The number of graphs varied in $k = 20, 40, 60, 80, 100$. The size of the graphs was set to $n=50$. This process was repeated 1,000 times for each number of graphs $k$.


\subsection{Results and analysis of the simulations}
One first natural question is, instead of using the spectral radius, may one use another feature of the graph, such as the number of edges or measures of network centrality? In order to verify whether the spectral radius is indeed better than those measures, we simulated the five random graph models described in section \ref{section:random} and compared the performance between the spectral radius against other six measures, namely transitivity centrality, betweenness centrality, closeness centrality, eigenvector centrality, degree centrality, and assortativity.

The experimental set-up is as follows. The parameters of the graphs were generated by uniform distributions: Erd\"os-R\'enyi random graph $p \sim U(0,1)$, random geometric graph $r \sim U(0,1)$, random regular graph $k \sim \text{integer part of }U(1,10)$, Barab\'asi-Albert random graph $p_s \sim \text{integer part of }U(1,4)$, Watts-Strogatz random graph $p_w \sim U(0,1)$. The number of graphs is set to $k=30$. The number of vertices of the graph varied in $n=25, 50, 75, 100$. Then, we calculated the Spearman's correlation coefficient between the original parameter used to generate the graph and the feature (the spectral radius, transitivity centrality, betweenness centrality, closeness centrality, eigenvector centrality, degree centrality, and assortativity). This process was repeated 100 times for each $n$ and graph model. The correlation calculated here measures how much information regarding the parameter is represented in the feature. In other words, it measures how well the feature describes the parameter. 

\begin{figure}
\centering
\includegraphics[width=4in]{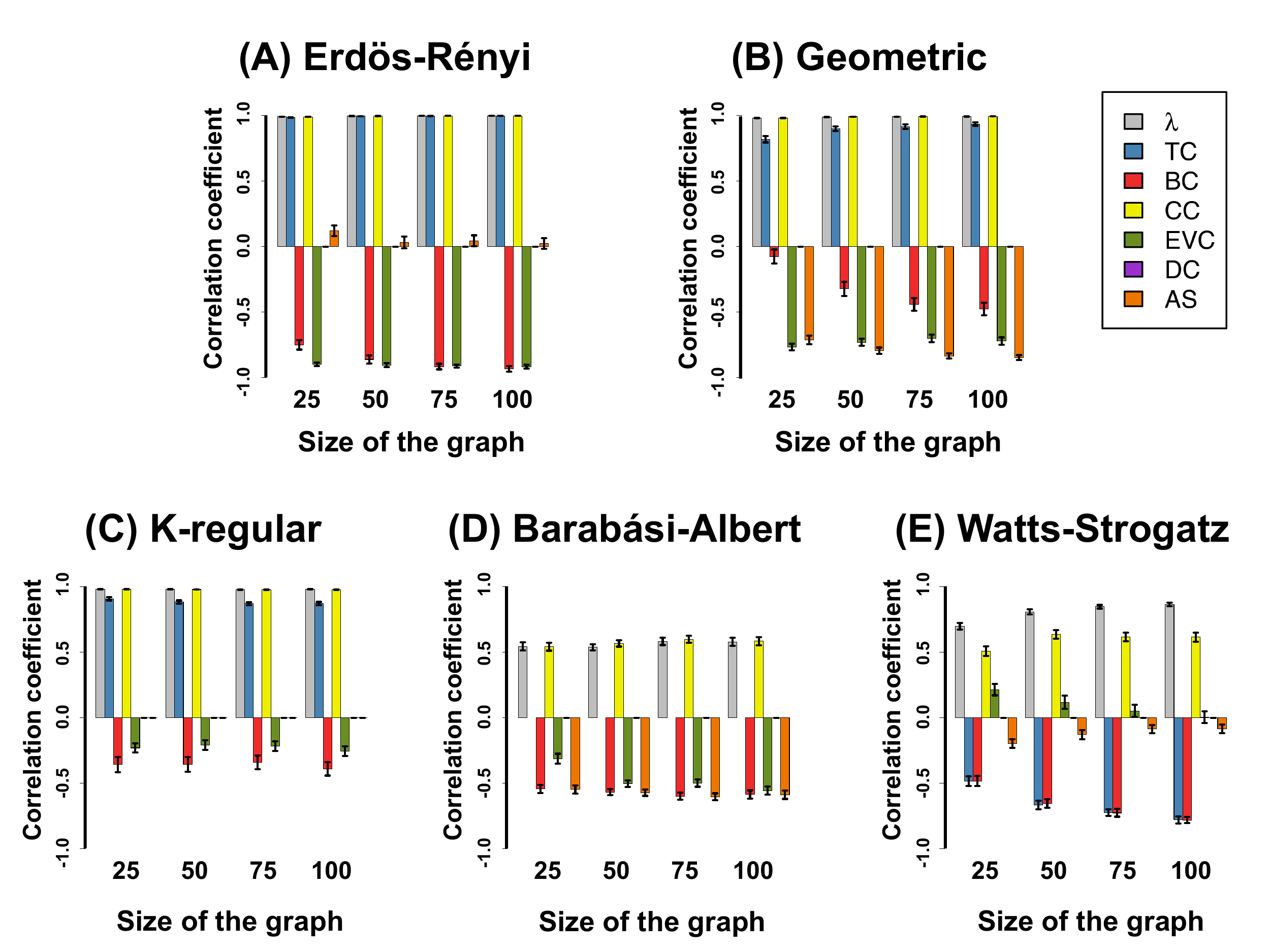}
\caption{Simulation study to select the most suited feature of the graph to be used in the identification of correlation between vectors of graphs. The $x$-axis and $y$-axis represent the size of the graph $n$ (number of vertices) and the average Spearman's correlation coefficient between the original parameter and the feature. The error bars represent the 95\% confidence interval. $\lambda$: spectral radius; TC: transitivity centrality; BC: betweenness centrality; CC: closeness centrality; EVC: eigenvector centrality; DC: degree centrality; AS: assortativity. Notice that for Erd\"os-R\'enyi, geometric, and $k$-regular graphs, spectral radius, transitivity, and closeness centralities are the best choices. For Albert-Barab\'asi random graph model, the spectral radius and the closeness centrality are the best features. However, for Watts-Strogatz random graph model, the spectral radius is the most correlated to the parameter ($p_w$) of the graph. \label{figure:best_feature}}
\end{figure}

Figure \ref{figure:best_feature} illustrates the average Spearman's correlation coefficient and the 95\% confidence interval between the actual parameter used to generate the graphs and the features. For Erd\"os-R\'enyi random graph, the spectral radius, transitivity, and closeness centralities presented the highest associations with the parameter $p$. For geometric, $k$-regular, and Barab\'asi-Albert random graph models, the spectral radius and the closeness centrality are the features that better represent the actual parameter. For Watts-Strogatz random graph model, the spectral radius is the most correlated to the parameter of the graph. By combining all these results, we conclude that the spectral radius is the one that contains the highest information regarding the parameters and consequently is the most suited feature to be used to identify correlation between graphs. The quite high performance of closeness centrality can be explained by the fact that for Erd\"os-R\'enyi, geometric, $k$-regular, and Barab\'asi-Albert random graph models, the parameters of the graph are associated with the number of edges. Notice that highly connected graphs tends to present higher closeness centrality. On the other hand, the parameter $p_w$ of the Watts-Strogatz random graph model represents the rewiring probability of the edges. In other words, what vary along Watts-Strogatz random graphs is their structure (connectivity), and not the number of edges. Thus, the performance of closeness centrality becomes poor and not adequate to identify correlation when the structure of the graph is modified without altering the number of edges. Some features such as the degree centrality and assortativity for the $k$-regular random graph model, and the transitivity centrality for the Barab\'asi-Albert random graph model could not be calculated. It happened because a $k$-regular graph presents the same degree for all vertices; therefore, all the vertices present the same degree centrality. For the Barab\'asi-Albert random graph model, a few vertices present very high degree while the majority of vertices present low degree.

Figure \ref{figure:boxplot} represents the correlation coefficient obtained by carrying out simulation 1 (section \ref{simulation1}). The $x$-axis indicates the real correlation used to generate the parameters while the boxplots on the $y$-axis indicate the correlation estimated by using the Spearman's correlation on the spectral radius. Notice that the correlation estimated by applying on the spectral radius ($y$-axis) is indeed monotonic (and also linear) in relation to the actual correlation between the parameters of the graph ($x$-axis). In other words, the higher the correlation between the parameters, the higher is the Spearman's correlation coefficient. Therefore, the spectral radius is in fact a good feature to identify correlation between graphs.

\begin{figure}
\begin{center}
\includegraphics[width=4in]{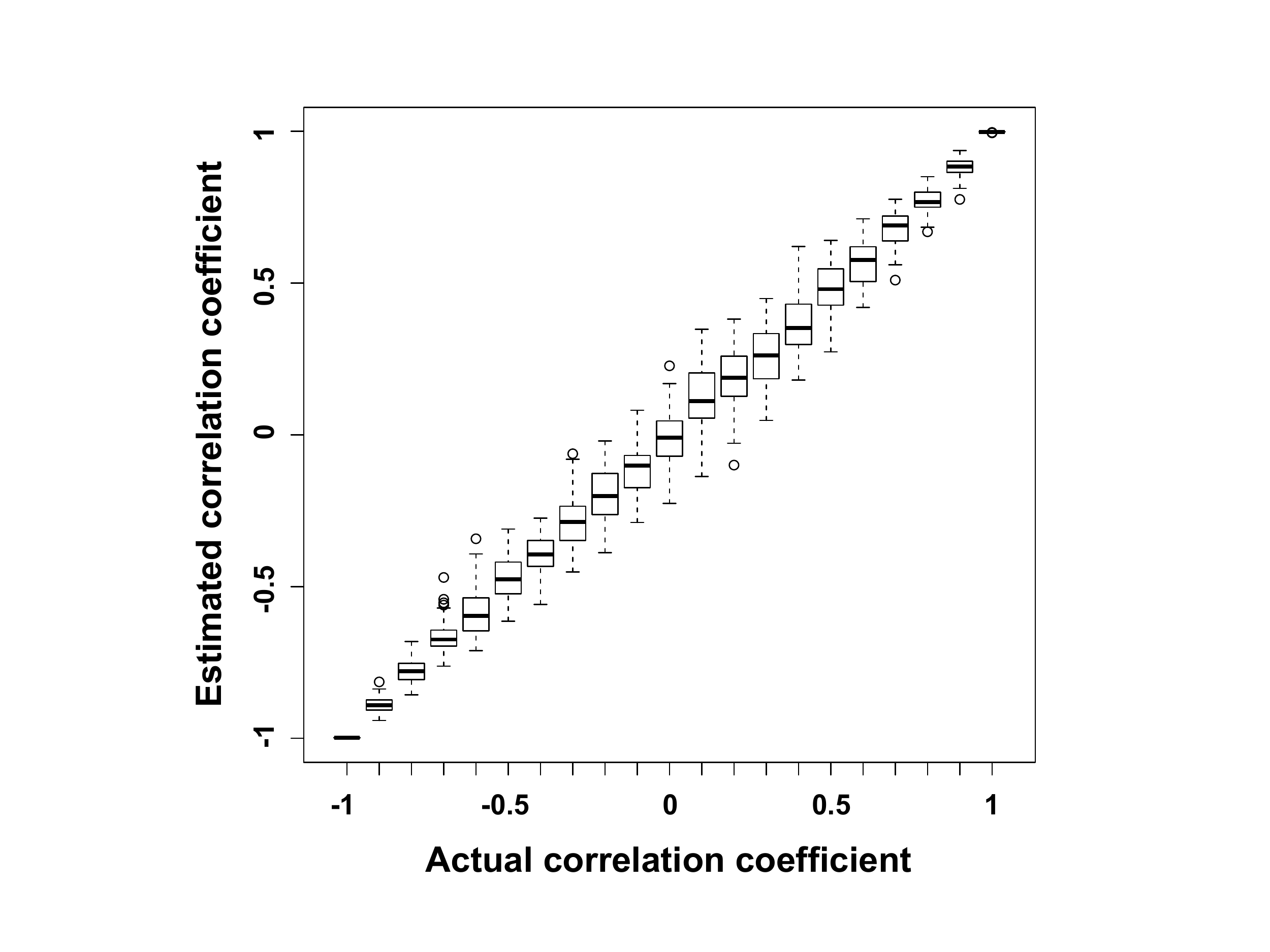}
\end{center}
\caption{Boxplots to investigate the monotonicity of the correlation between spectral radii and the real correlation between the parameters used to generate the graphs. The $x$-axis indicates the real correlation used to generate the parameters while the $y$-axis indicates the correlation estimated by the Spearman's correlation using the spectral radius. Notice that the estimated correlation is monotonic (and linear) in relation to the actual correlation between the parameters of the graphs. \label{figure:boxplot}}
\end{figure}

Figures \ref{figure:roc-H0} and \ref{figure:roc-H1} describe the ROC curves for the correlation between different classes of graph models (Erd\"os-R\'enyi, geometric, regular, Barab\'asi-Albert, and Watts-Strogatz) under the null and alternative hypotheses, respectively, in 1,000 repetitions. For further details regarding the design of this experiment, refer to simulation 2 described in section \ref{simulation2}. The panels on the upper triangle represent the ROCs curves obtained by using the spectral radius. The panel on the lower triangle represents the ROC curve (reference ROC curve) obtained by using the original parameters of the graph. By analyzing the ROC curves, it is possible to notice at least that: (i) the ROC curves under the null hypothesis are in the diagonal (Figure \ref{figure:roc-H0}), i.e., the statistical test is effectively controlling the rate of false positives (the proportion of rejected null hypothesis is as expected by the $p$-value threshold); (ii) the power of the test increases as the number of graphs ($k$) increases (Figure \ref{figure:roc-H1}); and (iii) ROC curves obtained by applying Spearman's correlation in the spectral radius (ROC curves in the upper triangle of Figure \ref{figure:roc-H0} and \ref{figure:roc-H1}) are similar to the one obtained by applying on the original parameters of the graphs (ROC curves in the lower triangle (reference ROC curve) of Figure \ref{figure:roc-H0} and \ref{figure:roc-H1}). These simulation studies show that, in fact, it is possible to retrieve the underlying correlation between the parameters of the graphs by analyzing the spectral radius of their adjacency matrices, at least, for these five random graph models.

\begin{figure}
\begin{center}
\includegraphics[width=4in]{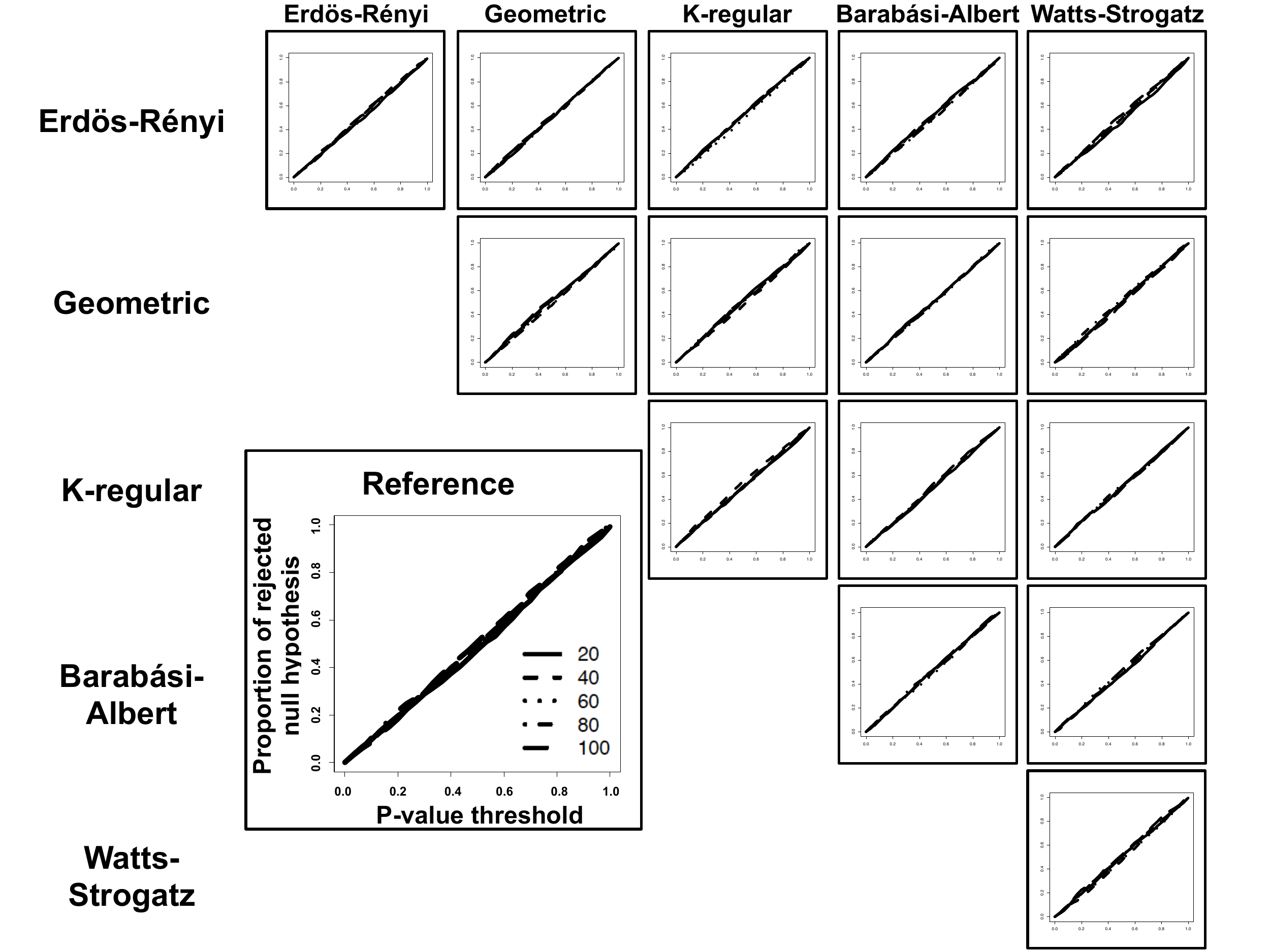}
\end{center}
\caption{On the upper triangle, ROC curves constructed based on the correlations estimated from spectral radii. On the lower triangle, the reference ROC curve constructed based on the parameter of the graph. The $x$-axis represents the $p$-value's threshold and the $y$-axis represents the proportion of rejected null hypothesis in 1,000 repetitions. The different types of line (solid and dashed) represent the number of graphs ($k=20, 40, 60, 80, 100$) used in each repetition. Notice that all lines are in the diagonal, i.e., the statistical test is indeed controlling the rate of false positives as expected. \label{figure:roc-H0}}
\end{figure}

\begin{figure}
\begin{center}
\includegraphics[width=4in]{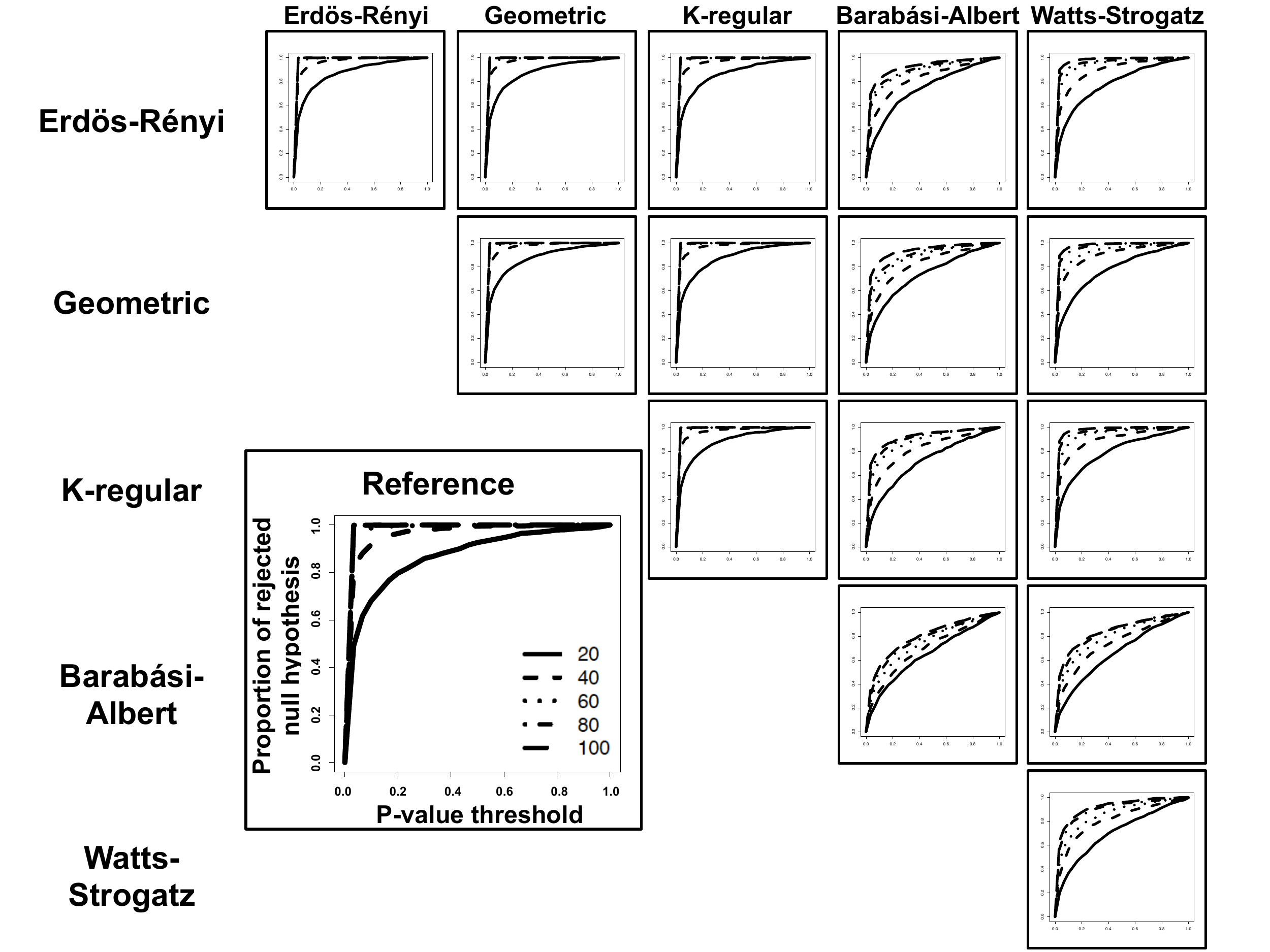}
\end{center}
\caption{On the upper triangle, ROC curves constructed based on the correlations estimated from spectral radii. On the lower triangle, the reference ROC curve constructed based on the parameters of the graph. The $x$-axis represents the $p$-value's threshold and the $y$-axis represents the proportion of rejected null hypothesis in 1,000 repetitions. The different types of line (solid and dashed) represent the number of graphs ($k=20, 40, 60, 80, 100$) used in each repetition. Notice that the greater the number of graphs ($k$), the higher is the power of the test. \label{figure:roc-H1}}
\end{figure}

Figure \ref{figure:fisher-roc} panels (A) and (B) describe the ROC curves for simulation 3 (simulation to evaluate the Fisher's test. For further details, refer to section \ref{simulation3}) under the null and alternative hypotheses, respectively. By analyzing Figure \ref{figure:fisher-roc}A, it is possible to notice that the test in fact controls the rate of false positives. By analyzing Figure \ref{figure:fisher-roc}B, it is possible to see that the power of the test increases proportionally to the number of graphs. In summary, the Fisher's test is indeed identifying distinct correlations between two conditions.

\begin{figure}
\begin{center}
\includegraphics[width=4in]{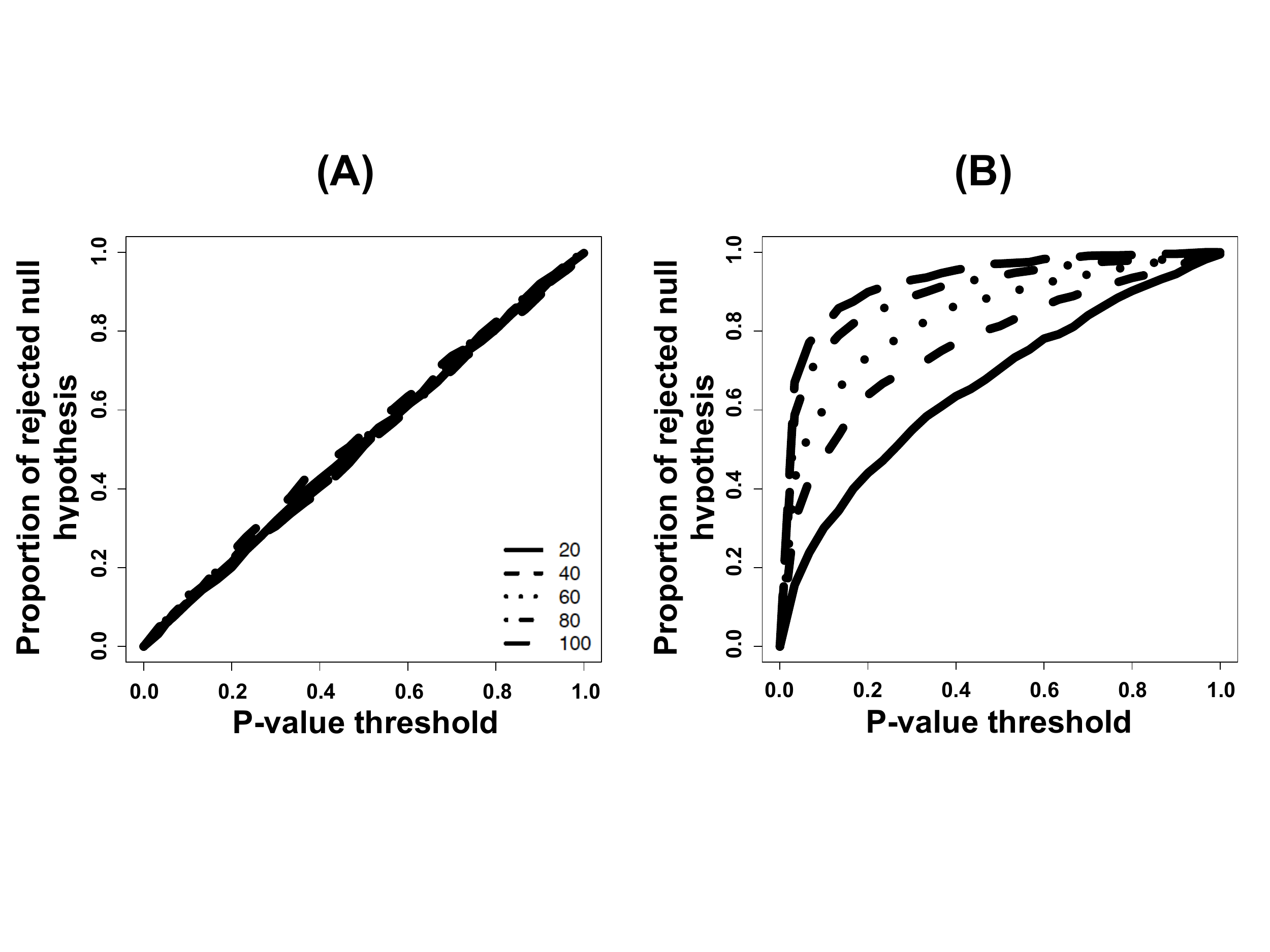}
\end{center}
\caption{Comparison of two conditions by using the Fisher's test. The $x$-axis represents the $p$-value's threshold and the $y$-axis represents the proportion of rejected null hypothesis in 1,000 repetitions. The different types of line (solid and dashed) represent the number of graphs ($k=20, 40, 60, 80, 100$) used in each repetition. (A) Under the null hypothesis, i.e., there is no difference in correlation between two conditions. The lines in the diagonal show that the statistical test is indeed controlling the rate of false positives. (B) The correlation between two conditions is in fact different. Notice that the greater the number of graphs ($k$), the higher is the power of the test. \label{figure:fisher-roc}}
\end{figure}


\section{Application to Autism Spectrum Disorder dataset}

Autism spectrum disorder (ASD) is a neurodevelopmental disorder usually diagnosed in early childhood. ASD etiology is complex and not completely understood (\cite{Ecker2015}), involving several risk factors, such as genetic, environmental, psychological, and neurobiological (\cite{Hallmayer, Betancur}). It is usually diagnosed by a multidisciplinary group composed of physicians and psychologists that, through clinical interviews and tests, identify a combination of unusual behavioral characteristics and try to assess deficits in social communication, social reciprocity, and repetitive and stereotyped behaviors and interests (\cite{Wing}). These symptoms frequently manifest during the child's first three years and are accompanied by developmental differences in brain anatomy, functioning, and functional brain connectivity.

Current studies suggest that ASD is a disorder of brain systems (\cite{Wass2011, Stevenson2012, Just2012, Frith2003})  and that  anatomical abnormalities are subtle but widespread over the brain (\cite{Ecker2013}). Thus, one straightforward approach to enhance our comprehension of neural substrates of this disorder is to investigate differences in brain connectivity when compared to controls. In this context, most studies focus on finding differences between region-to-region functional connectivity or in vertex centrality measures. Due to the lack of a suitable methodological framework, investigations in how the structural organization in one brain sub-network is associated with the organization of another sub-network is scarce. Moreover, the description of these ``correlations'' among sub-networks in clinical populations remains unexplored. In the current study, we establish a novel framework to define correlation between graph structures and illustrate the usefulness of this method by enhancing our comprehension on the neurobiology of ASD. 


\subsection{Dataset description}
A large resting state fMRI dataset initially composed of 908 individuals comprising controls and subjects diagnosed with ASD was downloaded from the ABIDE Consortium website (\url{http://fcon_1000.projects.nitrc.org/indi/abide/}). The ABIDE dataset is fully anonymized in compliance with the HIPAA Privacy Rules and the 1000 Functional Connectomes Project/INDI protocols. Protected health information is not included in this dataset. Further details can be obtained from the ABIDE Consortium website.

The pre-processing of the imaging data was performed using the Athena pipeline downloaded from (\url{http://www.nitrc.org/plugins/mwiki/index.php/neurobureau:AthenaPipeline}). The 351 regions of interest (ROIs) considered as the vertices of the brain network were defined by the CC400 atlas (\cite{Craddock}). A total of 35 ROIs including the ventricles were identified by using the MNI atlas and removed, resulting 316 ROIs for the construction of brain networks. The average time series within the ROIs were considered as to be the region representatives. Subject's head movement during magnetic resonance scanning was treated by using the ``scrubbing'' procedure described by \cite{Power}. Individuals with a number of adequate scans less than 100 after the ``scrubbing'' were discarded, which resulted in 814 subjects for subsequent analyses. Thus, the dataset used in this study was composed of 529 controls (430 males, mean age $\pm$ standard deviation, $17.47 \pm 7.81$ years) and 285 ASD (255 males, $17.53 \pm 7.13$ years).


\subsection{Brain functional networks}

The schema of the entire fMRI data analysis can be seen in Figure \ref{figure:schema}. A brain functional network can be modeled as a graph, i.e., a pair of sets $G=(V,E)$, in which $V$ is the set of regions of interest - ROIs (vertices), and $E$ is a set of functional connectivity (edges) connecting the ROIs. In the current study, the functional connectivity between two ROIs was obtained by calculating the Spearman's correlation coefficient between ROIs $i$ and $j$ ($i,j = 1, \ldots, 316$) for each individual $q=1, \ldots, 814$. Thus, a brain functional network $G^q$ with 316 ROIs can be represented by its adjacency matrix ${\bf A}^q$ with $316 \times 316$ elements ${\bf A}^q_{ij}$ containing the connectivity (Spearman's correlation coefficient) between the ROIs $i$ and $j$ ($i,j=1,\ldots, 316$; $q=1, \ldots, 814$). 

\begin{figure}
\begin{center}
\includegraphics[width=4in]{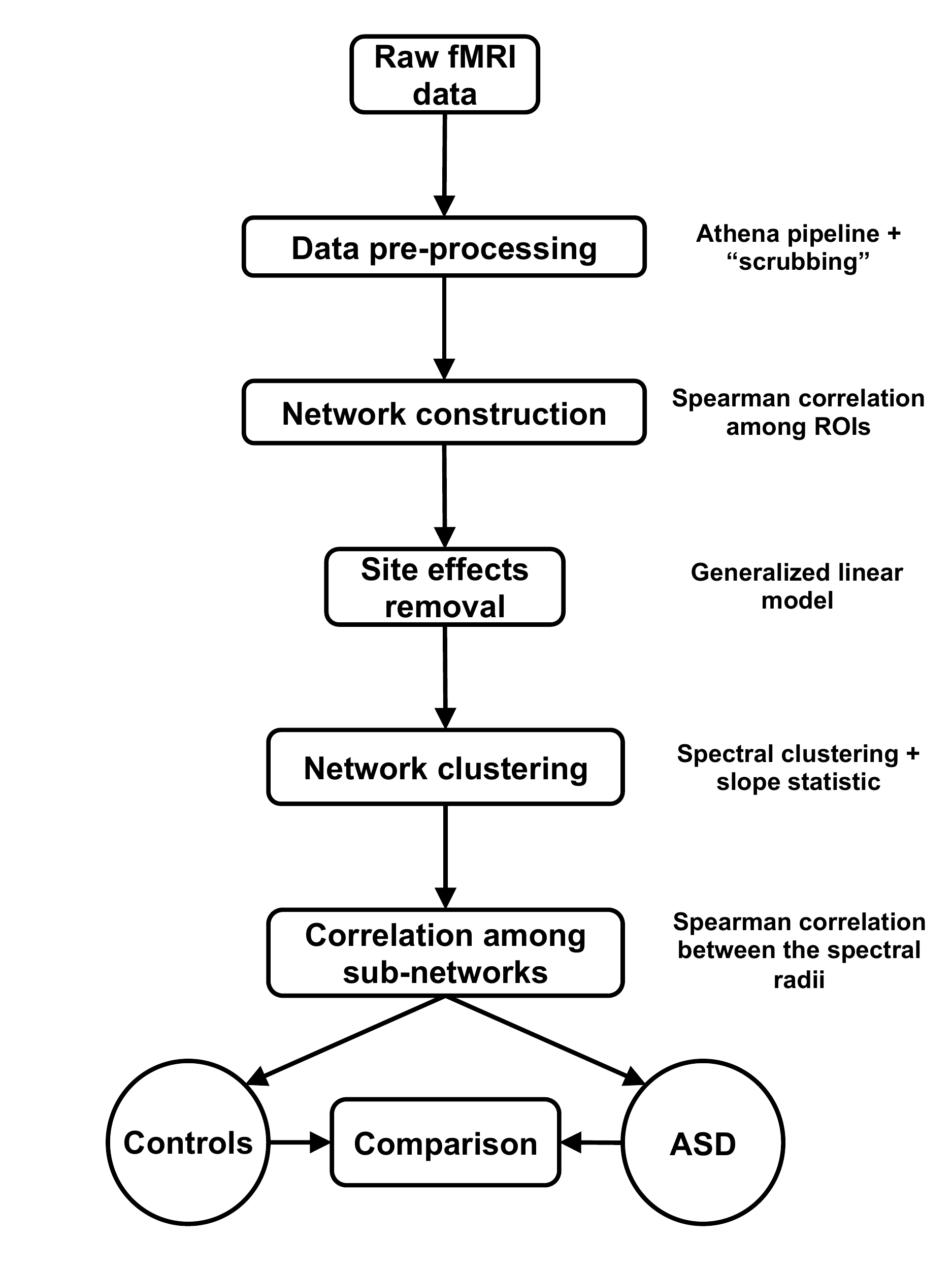}
\end{center}
\caption{General pipeline schema of the fMRI data analysis. Raw fMRI data is pre-processed by using the Athena pipeline. Head movement effects are removed by using the “scrubbing” procedure. Functional networks were constructed by estimating the Spearman's correlation among ROIs and site effects removed by a generalized linear model. The number of functional brain sub-networks (clusters) were estimated by the slope statistic and the sub-networks were obtained by applying the spectral clustering algorithm. Finally, the correlation among sub-networks were estimated by the Spearman's correlation applied on the spectral radii. \label{figure:schema}}
\end{figure}

Site effects were modeled with a generalized linear model (GLM), i.e., with the site as a categorical predictor variable and the correlation coefficient as the response variable. The residuals of the model were used for subsequent analyses as the connectivity filtered by the site effect.

$P$-values for each Spearman's correlation coefficient between ROIs $i$ and $j$ were calculated and corrected for the false discovery rate (FDR) (\cite{Benjamini}). To estimate the number of sub-networks and also to identify the sub-networks themselves in a data-driven manner, we applied the slope criterion (\cite{Fujita}) and the spectral clustering algorithm (\cite{Ng}), respectively, on the average connectivity matrix (the average of the $z$-values associated with the $p$-values) taking into account the entire dataset. The application of the spectral clustering resulted in five (estimated by the slope statistic) well defined sub-networks namely somatomotor, visual, default-mode, cerebellar, and control, depicted in Figure \ref{figure:cluster}. 

\begin{figure}
\begin{center}
\includegraphics[width=4in]{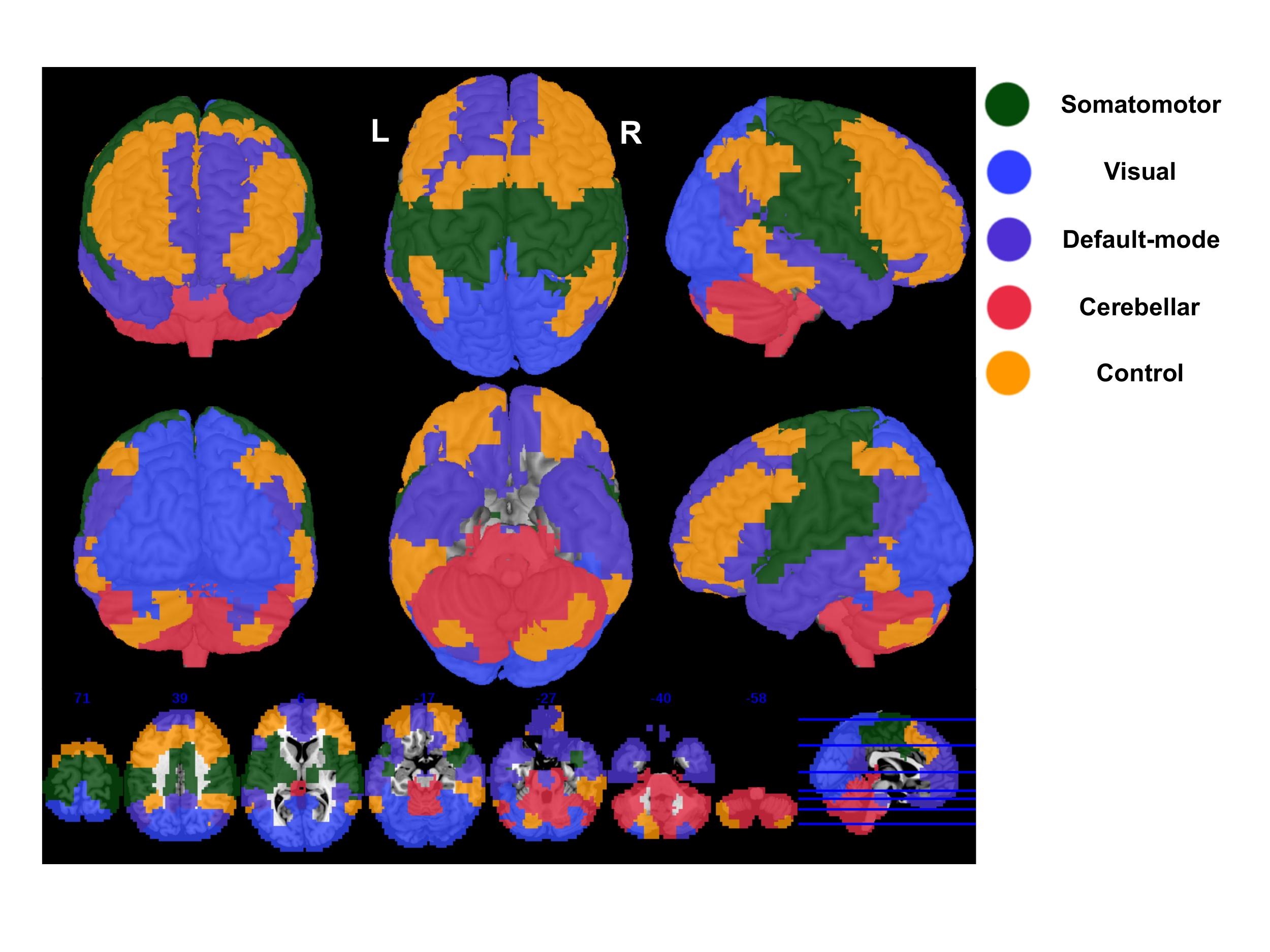}
\end{center}
\caption{Brain functional sub-networks. The ROIs were clustered by the spectral clustering algorithm. The number of sub-networks was estimated as five by the slope statistic. Each sub-network is represented by a different color, namely somatomotor (green), visual (blue), default-mode (purple), cerebellar (red), and control (orange). R: right; L: left. \label{figure:cluster}}
\end{figure}

To obtain the adjacency matrix that represents the brain functional sub-network of each individual, we set ${\bf A}^{q}_{ij} = 1$ if the $p$-value corrected for the FDR is less than 0.05, and ${\bf A}^{q}_{ij} = 0$, otherwise. Notice that the Spearman's correlation test is not used as a statistical test to identify correlation between two ROIs but only as an objective criterion to construct the adjacency matrix of the graph. Then, we calculated the spectral radius for each sub-network of each subject. Thus, we obtained five vectors (one for each sub-network) of size 529 and other five vectors of size 285, for controls and ASD, respectively.

The 10 correlations among all the five sub-networks were estimated by using the spectral radius only for controls (Figure \ref{figure:net}). Figure \ref{figure:net} shows the statistically significant ($p<0.05$ after FDR correction for multiple tests) correlations between brain sub-networks. Interestingly, all sub-networks are positively correlated among them. The thickness of the edge represents the strength of the correlation, i.e., the thicker the edge, the higher is the correlation between sub-networks.

\begin{figure}
\begin{center}
\includegraphics[width=4in]{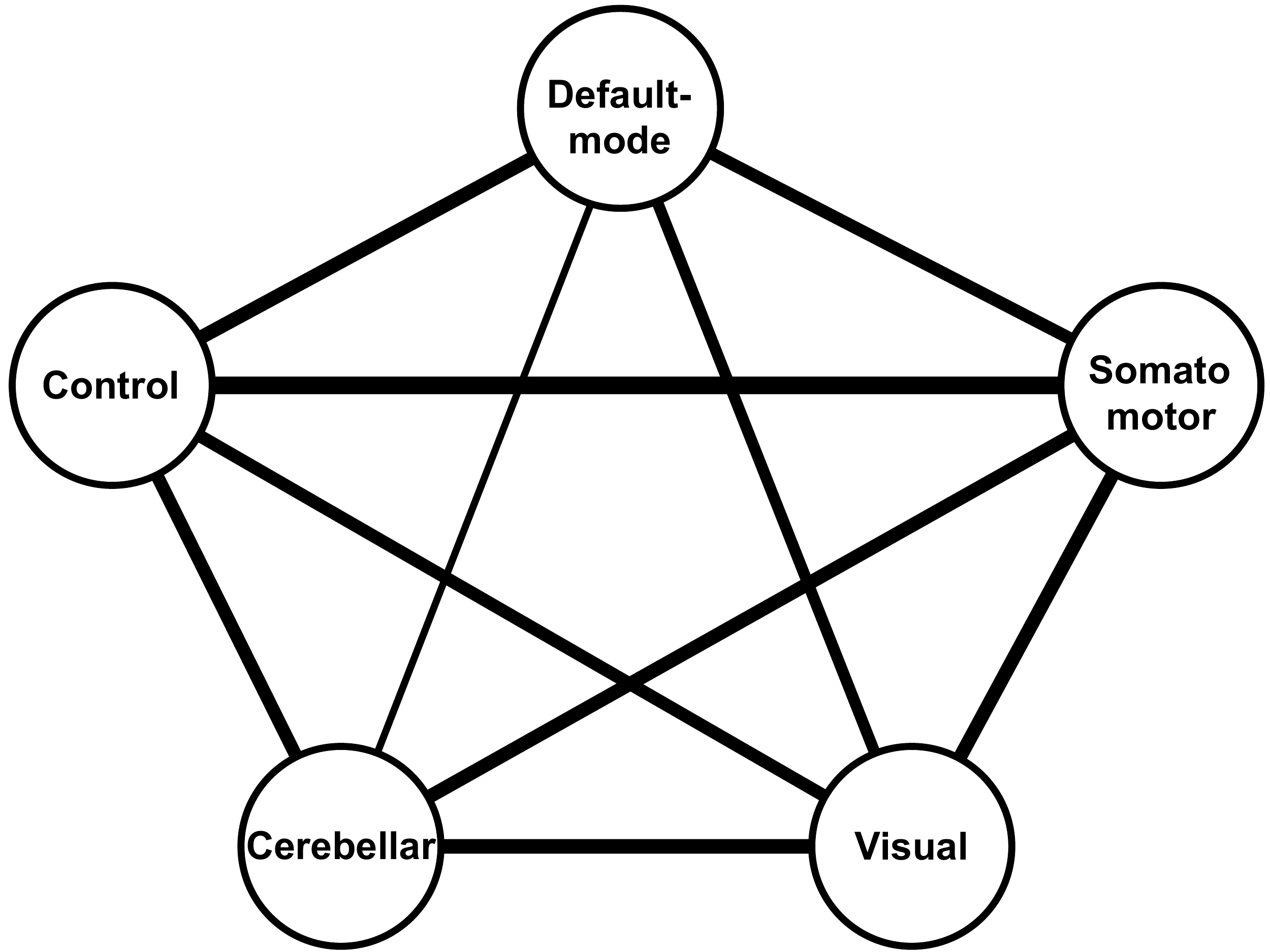}
\end{center}
\caption{Correlation between control brain sub-networks. The thickness of the edge represents the correlation coefficient, i.e., the thicker the edge, the higher is the absolute value of the correlation between sub-networks. Only statistically significant correlations at a significance threshold of 0.05 (after FDR correction for multiple tests) are shown. All identified correlations are positive. \label{figure:net}}
\end{figure}

Then, to identify correlations among sub-networks that are different between controls and ASD, we carried out the statistical test developed by Fisher (\cite{Fisher}) (section \ref{section:spearman_correlation}).

Figure \ref{figure:diff-net} illustrates the results obtained by comparing the correlations between controls versus ASD in the ``scrubbed'' data. The ``scrubbing'' procedure (\cite{Power}) is necessary to remove head movement effects that may cause spurious results. Only correlations that are statistically different between controls and ASD are represented by edges. At the edges there are the scatterplots with the linear regression lines that fit the data (blue and red represent the spectral radii of controls and ASD, respectively) and the Spearman's correlation coefficient for controls and ASD. Interestingly, default-mode and control, default-mode and somatomotor, and default-mode and visual systems showed statistically significant higher inter-correlation in ASD when compared to controls ($p<0.05$ after FDR correction for multiple tests).

\begin{figure}
\centering
\includegraphics[width=4in]{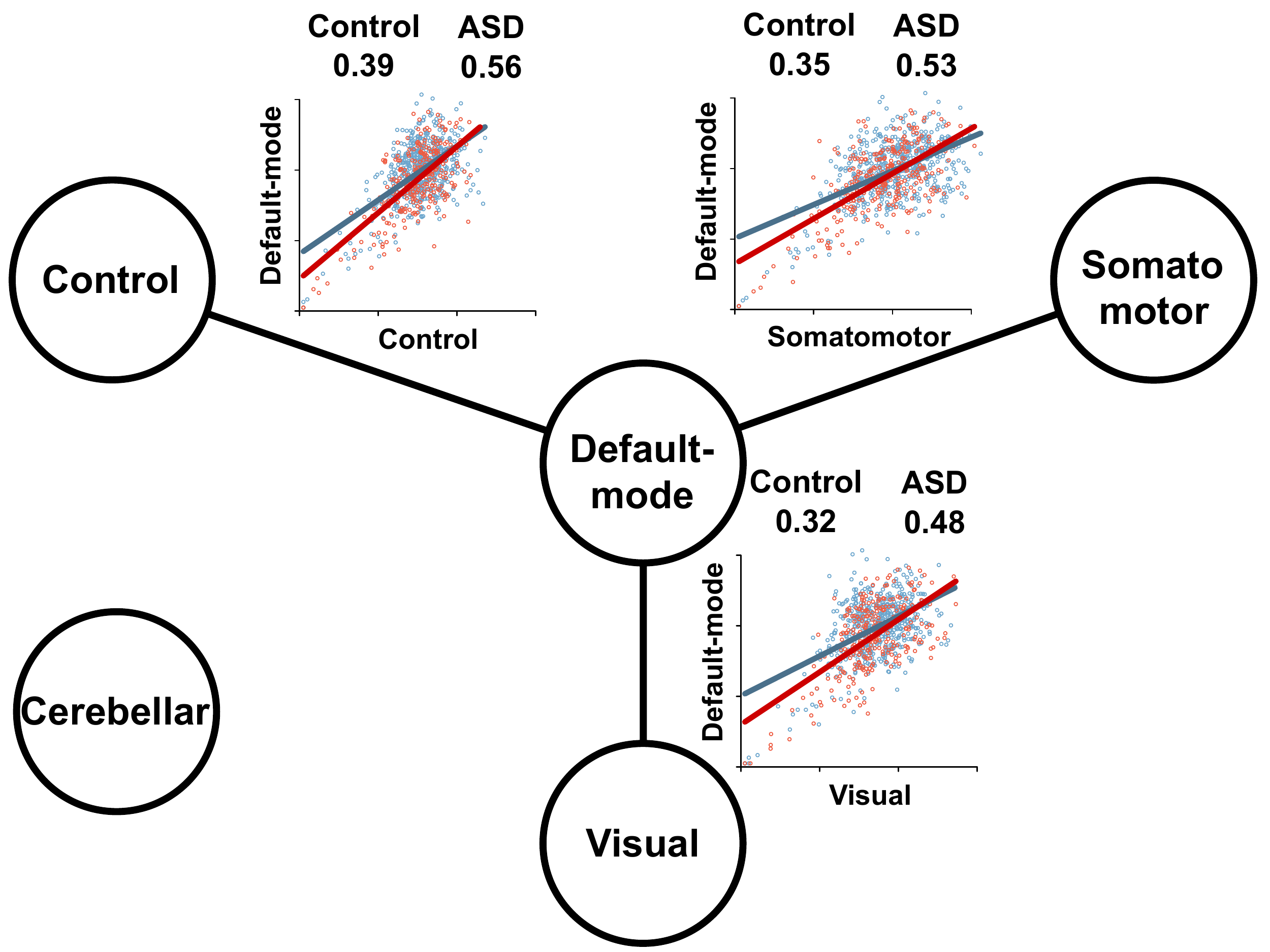}
\caption{Correlation between brain sub-networks in the ``scrubbed'' data. At the edges there are represented the scatterplots with the linear regression lines that fit the data and the Spearman's correlation coefficient for controls and ASD. Blue and red dots represent the eigenvalues for controls and ASD, respectively. Default-mode and control, default-mode and somatomotor, and default-mode and visual systems showed higher inter-correlation in ASD than in controls ($p<0.05$ after FDR correction for multiple tests) \label{figure:diff-net}}
\end{figure}

Currently, ASD is characterized as a disconnection syndrome (\cite{geschwind2007autism}) that affects information processing at both the local and global levels (\cite{hernandez2014neural}). This hypothesis has been mainly derived from studies suggesting that the wide heterogeneity of autistic symptoms and traits are highly unlikely to be due to impairments in a single system,  or brain region, but instead emerge from disruptions in multiple neurocognitive  systems (for a review see \cite{Ecker2015}). The framework for correlation inference between graphs proposed in our study therefore represents a useful and suitable method for examining brain connectivity alterations in ASD related to abnormalities in the relationships between domain-specific sub-networks in the brain.

Our findings indicate a different pattern of interactions between the default-mode and several sub-networks associated with sensorymotor, visual and executive processing in ASD. Abnormalities in the connectivity between nodes of the default-mode network (DMN) has been widely investigated in ASD (\cite{assaf2010abnormal, kennedy2008functional, weng2010alterations}) giving its associations with social cognition (\cite{buckner2008brain}). There are also functional connectivity studies reporting ASD-related differences in motor and visual networks (\cite{mcgrath2012atypical, nebel2014precentral}). Moreover, task-based connectivity studies have also reported differences in individuals with ASD in fronto-parietal nodes of the control network in response to cognitive control tasks (\cite{kana2007inhibitory}). While these previous studies have contributed to the characterization of the disconnection model of ASD, evidence of commonly interacting abnormalities in several distinct neurocognitive systems remain scarce. The present results therefore go beyond previous research by providing evidence for altered interactions between sub-networks in ASD ranging from sensory and motor processing to higher-order cognitive functions.

The mechanisms underlying this pattern of widespread connectivity abnormalities are likely to be related to the complex nature of atypical trajectory of brain maturation in ASD, which is probably mediated by many genetic and environmental factors and their interactions (\cite{abrahams2010connecting}). As a result, the impact of ASD on brain anatomy, functioning and connectivity is expected to be multidimensional and observable at multiple neural systems (\cite{Ecker2015}). For instance, ASD prediction accuracy using a support vector machine (SVM) analytic approach (\cite{ecker2010describing}) was shown to be improved by the use of a combined set of different morphometric features of the cortical surface rather than a specific anatomical characteristic; moreover a spatially distributed pattern of regions instead of an isolated brain region contributed with maximal classification weights to the prediction model. It is therefore likely that the distributed patterns of functional connectivity differences reported here also reflect some of these systems-level features of ASD pathology.

This result reinforces the theory that ASD is related to abnormal neurodevelopmental processes (Frith, 2003), which is spread across the whole brain. By using the proposed methodological framework of graph correlations, our main contribution was to demonstrate that these latent processes result in correlated topological organization between multiple brain systems. To the best of our knowledge, this is the first description of this feature in ASD.


\section{Final remarks}
In this study, we used the Spearman's rank correlation as a measure to identify dependence between graphs due to its simplicity. However, instead of Spearman's rank correlation, other methods that identify a broader type of dependence in data can be used, such as distance correlation (\cite{Szekely2007}), mutual information (\cite{shannon2015mathematical}), Hoeffding's D measure (\cite{hoeffding1948non}), and the measure proposed by Heller-Heller-Gorfine (\cite{Heller}).

Here, we focused on the development of a method to identify correlation for undirected graphs. For directed graphs, little is known about their spectrum. Notice that the adjacency matrix of an undirected graph is not symmetric, and consequently, the eigenvalues are not real numbers. Thus, in order to develop a framework to identify correlation between undirected graphs, it is necessary to better understand how are their eigenvalues (or eventually use another approach not based on their spectrum).

To the best of our knowledge, there is no general analytical equation that describes the spectral radius as a function of the parameters for all kinds of random graphs models. However, for the Erd\"os-R\'enyi random graph model we prove that the proposed approach is consistent and not biased. For the other four random graph models used in this study, we showed by simulations that the spectral radius is indeed associated with the parameters of the graph and it is also a good feature to infer correlation between graphs. This approach based on the spectral radius seems to be promising and we hope it may open opportunities to develop other formal statistical methods in graphs.

\bibliography{main}

\begin{thebibliography}{10}

\bibitem{abrahams2010connecting}
Brett~S Abrahams and Daniel~H Geschwind.
\newblock Connecting genes to brain in the autism spectrum disorders.
\newblock {\em Archives of neurology}, 67(4):395--399, 2010.

\bibitem{Alon}
Noga Alon.
\newblock Eigenvalues and expanders.
\newblock {\em Combinatorica}, 6(2):83--96, 1986.

\bibitem{assaf2010abnormal}
Michal Assaf, Kanchana Jagannathan, Vince~D Calhoun, Laura Miller, Michael~C
  Stevens, Robert Sahl, Jacqueline~G O'Boyle, Robert~T Schultz, and Godfrey~D
  Pearlson.
\newblock Abnormal functional connectivity of default mode sub-networks in
  autism spectrum disorder patients.
\newblock {\em Neuroimage}, 53(1):247--256, 2010.

\bibitem{Barabasi}
Albert-L{\'a}szl{\'o} Barab{\'a}si and R{\'e}ka Albert.
\newblock Emergence of scaling in random networks.
\newblock {\em science}, 286(5439):509--512, 1999.

\bibitem{Benjamini}
Yoav Benjamini and Yosef Hochberg.
\newblock Controlling the false discovery rate: a practical and powerful
  approach to multiple testing.
\newblock {\em Journal of the Royal Statistical Society. Series B
  (Methodological)}, 57(1):289--300, 1995.

\bibitem{Betancur}
Catalina Betancur.
\newblock Etiological heterogeneity in autism spectrum disorders: more than 100
  genetic and genomic disorders and still counting.
\newblock {\em Brain research}, 1380:42--77, 2011.

\bibitem{Boccaletti}
Stefano Boccaletti, Vito Latora, Yamir Moreno, Martin Chavez, and D-U Hwang.
\newblock Complex networks: Structure and dynamics.
\newblock {\em Physics reports}, 424(4):175--308, 2006.

\bibitem{Bordenave}
Charles Bordenave.
\newblock Eigenvalues of euclidean random matrices.
\newblock {\em Random Structures \& Algorithms}, 33(4):515--532, 2008.

\bibitem{borkowf2002computing}
Craig~B Borkowf.
\newblock Computing the nonnull asymptotic variance and the asymptotic relative
  efficiency of spearman's rank correlation.
\newblock {\em Computational statistics \& data analysis}, 39(3):271--286,
  2002.

\bibitem{buckner2008brain}
Randy~L Buckner, Jessica~R Andrews-Hanna, and Daniel~L Schacter.
\newblock The brain's default network.
\newblock {\em Annals of the New York Academy of Sciences}, 1124(1):1--38,
  2008.

\bibitem{Craddock}
R~Cameron Craddock, G~Andrew James, Paul~E Holtzheimer, Xiaoping~P Hu, and
  Helen~S Mayberg.
\newblock A whole brain fmri atlas generated via spatially constrained spectral
  clustering.
\newblock {\em Human brain mapping}, 33(8):1914--1928, 2012.

\bibitem{ding2010spectral}
Xue Ding, Tiefeng Jiang, et~al.
\newblock Spectral distributions of adjacency and laplacian matrices of random
  graphs.
\newblock {\em The annals of applied probability}, 20(6):2086--2117, 2010.

\bibitem{Dorogovtsev}
S.~N. Dorogovtsev, A.~V. Goltsev, J.~F.~F. Mendes, and A.~N. Samukhin.
\newblock Spectra of complex networks.
\newblock {\em Phys. Rev. E}, 68:046109, Oct 2003.

\bibitem{Ecker2013}
C~Ecker, W~Spooren, and DGM Murphy.
\newblock Translational approaches to the biology of autism: false dawn or a
  new era?
\newblock {\em Molecular psychiatry}, 18(4):435--442, 2013.

\bibitem{Ecker2015}
Christine Ecker, Susan~Y Bookheimer, and Declan~GM Murphy.
\newblock Neuroimaging in autism spectrum disorder: brain structure and
  function across the lifespan.
\newblock {\em The Lancet Neurology}, 2015.

\bibitem{ecker2010describing}
Christine Ecker, Andre Marquand, Janaina Mour{\~a}o-Miranda, Patrick Johnston,
  Eileen~M Daly, Michael~J Brammer, Stefanos Maltezos, Clodagh~M Murphy, Dene
  Robertson, Steven~C Williams, et~al.
\newblock Describing the brain in autism in five dimensions—magnetic
  resonance imaging-assisted diagnosis of autism spectrum disorder using a
  multiparameter classification approach.
\newblock {\em The Journal of Neuroscience}, 30(32):10612--10623, 2010.

\bibitem{Erdos}
P~Erd\"os and A~R\'enyi.
\newblock On random graphs i.
\newblock {\em Publ. Math. Debrecen}, 6:290--297, 1959.

\bibitem{Fisher}
RA~Fisher.
\newblock On the ``probable error'' of a coefficient of correlation deduced
  from a small sample.
\newblock {\em Metron}, 1(Pt 4):1--32, 1921.

\bibitem{Frith2003}
Chris Frith.
\newblock What do imaging studies tell us about the neural basis of autism.
\newblock {\em Autism: Neural basis and treatment possibilities}, pages
  149--176, 2003.

\bibitem{Fujita}
Andr{\'e} Fujita, Daniel~Y Takahashi, and Alexandre~G Patriota.
\newblock A non-parametric method to estimate the number of clusters.
\newblock {\em Computational Statistics \& Data Analysis}, 73:27--39, 2014.

\bibitem{Furedi}
Zolt{\'a}n F{\"u}redi and J{\'a}nos Koml{\'o}s.
\newblock The eigenvalues of random symmetric matrices.
\newblock {\em Combinatorica}, 1(3):233--241, 1981.

\bibitem{geschwind2007autism}
Daniel~H Geschwind and Pat Levitt.
\newblock Autism spectrum disorders: developmental disconnection syndromes.
\newblock {\em Current opinion in neurobiology}, 17(1):103--111, 2007.

\bibitem{Hallmayer}
Joachim Hallmayer, Sue Cleveland, Andrea Torres, Jennifer Phillips, Brianne
  Cohen, Tiffany Torigoe, Janet Miller, Angie Fedele, Jack Collins, Karen
  Smith, et~al.
\newblock Genetic heritability and shared environmental factors among twin
  pairs with autism.
\newblock {\em Archives of general psychiatry}, 68(11):1095--1102, 2011.

\bibitem{Heller}
Ruth Heller, Yair Heller, and Malka Gorfine.
\newblock A consistent multivariate test of association based on ranks of
  distances.
\newblock {\em Biometrika}, 2012.

\bibitem{hernandez2014neural}
Leanna~M Hernandez, Jeffrey~D Rudie, Shulamite~A Green, Susan Bookheimer, and
  Mirella Dapretto.
\newblock Neural signatures of autism spectrum disorders: insights into brain
  network dynamics.
\newblock {\em Neuropsychopharmacology}, 2014.

\bibitem{hoeffding1948non}
Wassily Hoeffding.
\newblock A non-parametric test of independence.
\newblock {\em The Annals of Mathematical Statistics}, pages 546--557, 1948.

\bibitem{Ingalhalikar2014}
Madhura Ingalhalikar, Alex Smith, Drew Parker, Theodore~D Satterthwaite, Mark~A
  Elliott, Kosha Ruparel, Hakon Hakonarson, Raquel~E Gur, Ruben~C Gur, and
  Ragini Verma.
\newblock Sex differences in the structural connectome of the human brain.
\newblock {\em Proceedings of the National Academy of Sciences},
  111(2):823--828, 2014.

\bibitem{Just2012}
Marcel~Adam Just, Timothy~A Keller, Vicente~L Malave, Rajesh~K Kana, and
  Sashank Varma.
\newblock Autism as a neural systems disorder: a theory of frontal-posterior
  underconnectivity.
\newblock {\em Neuroscience \& Biobehavioral Reviews}, 36(4):1292--1313, 2012.

\bibitem{kana2007inhibitory}
Rajesh~K Kana, Timothy~A Keller, Nancy~J Minshew, and Marcel~Adam Just.
\newblock Inhibitory control in high-functioning autism: decreased activation
  and underconnectivity in inhibition networks.
\newblock {\em Biological psychiatry}, 62(3):198--206, 2007.

\bibitem{kennedy2008functional}
Daniel~P Kennedy and Eric Courchesne.
\newblock Functional abnormalities of the default network during self-and
  other-reflection in autism.
\newblock {\em Social cognitive and affective neuroscience}, 3(2):177--190,
  2008.

\bibitem{mcgrath2012atypical}
Jane McGrath, Katherine Johnson, Christine Ecker, Erik O'Hanlon, Michael Gill,
  Louise Gallagher, and Hugh Garavan.
\newblock Atypical visuospatial processing in autism: insights from functional
  connectivity analysis.
\newblock {\em Autism Research}, 5(5):314--330, 2012.

\bibitem{Meringer}
Markus Meringer.
\newblock Fast generation of regular graphs and construction of cages.
\newblock {\em Journal of Graph Theory}, 30(2):137--146, 1999.

\bibitem{nebel2014precentral}
Mary~Beth Nebel, Ani Eloyan, Anita~D Barber, and Stewart~H Mostofsky.
\newblock Precentral gyrus functional connectivity signatures of autism.
\newblock {\em Frontiers in systems neuroscience}, 8, 2014.

\bibitem{Newman}
Mark~EJ Newman.
\newblock Modularity and community structure in networks.
\newblock {\em Proceedings of the National Academy of Sciences},
  103(23):8577--8582, 2006.

\bibitem{Ng}
Andrew~Y Ng, Michael~I Jordan, Yair Weiss, et~al.
\newblock On spectral clustering: Analysis and an algorithm.
\newblock {\em Advances in neural information processing systems}, 2:849--856,
  2002.

\bibitem{Penrose}
Mathew Penrose.
\newblock {\em Random geometric graphs}, volume~5.
\newblock Oxford University Press Oxford, 2003.

\bibitem{Perry2015}
Alistair Perry, Wei Wen, Anton Lord, Anbupalam Thalamuthu, Gloria Roberts,
  Philip~B Mitchell, Perminder~S Sachdev, and Michael Breakspear.
\newblock The organisation of the elderly connectome.
\newblock {\em NeuroImage}, 2015.

\bibitem{Power}
Jonathan~D Power, Kelly~A Barnes, Abraham~Z Snyder, Bradley~L Schlaggar, and
  Steven~E Petersen.
\newblock Spurious but systematic correlations in functional connectivity mri
  networks arise from subject motion.
\newblock {\em Neuroimage}, 59(3):2142--2154, 2012.

\bibitem{Rubinov}
Mikail Rubinov and Olaf Sporns.
\newblock Complex network measures of brain connectivity: uses and
  interpretations.
\newblock {\em Neuroimage}, 52(3):1059--1069, 2010.

\bibitem{shannon2015mathematical}
Claude~E Shannon and Warren Weaver.
\newblock {\em The mathematical theory of communication}.
\newblock University of Illinois press, 2015.

\bibitem{Sherwin2015}
Jason~Samuel Sherwin, Jordan Muraskin, and Paul Sajda.
\newblock Pre-stimulus functional networks modulate task performance in
  time-pressured evidence gathering and decision-making.
\newblock {\em NeuroImage}, 2015.

\bibitem{spearman1904general}
Charles Spearman.
\newblock " general intelligence," objectively determined and measured.
\newblock {\em The American Journal of Psychology}, 15(2):201--292, 1904.

\bibitem{Stam2014}
Cornelis~J Stam.
\newblock Modern network science of neurological disorders.
\newblock {\em Nature Reviews Neuroscience}, 15(10):683--695, 2014.

\bibitem{Stevenson2012}
Ryan~A Stevenson.
\newblock Using functional connectivity analyses to investigate the bases of
  autism spectrum disorders and other clinical populations.
\newblock {\em The Journal of Neuroscience}, 32(50):17933--17934, 2012.

\bibitem{Strogatz}
Steven~H Strogatz.
\newblock Exploring complex networks.
\newblock {\em Nature}, 410(6825):268--276, 2001.

\bibitem{Szekely2007}
Gábor~J. Székely, Maria~L. Rizzo, and Nail~K. Bakirov.
\newblock Measuring and testing dependence by correlation of distances.
\newblock {\em Ann. Statist.}, 35(6):2769--2794, 12 2007.

\bibitem{Mieghem}
Piet Van~Mieghem.
\newblock {\em Graph spectra for complex networks}.
\newblock Cambridge University Press, 2010.

\bibitem{Wass2011}
Sam Wass.
\newblock Distortions and disconnections: disrupted brain connectivity in
  autism.
\newblock {\em Brain and cognition}, 75(1):18--28, 2011.

\bibitem{Watts}
D.J. Watts and S.H. Strogatz.
\newblock Collective dynamics of 'small-world' networks.
\newblock {\em Nature}, 393:440--442, 1998.

\bibitem{weng2010alterations}
Shih-Jen Weng, Jillian~Lee Wiggins, Scott~J Peltier, Melisa Carrasco, Susan
  Risi, Catherine Lord, and Christopher~S Monk.
\newblock Alterations of resting state functional connectivity in the default
  network in adolescents with autism spectrum disorders.
\newblock {\em Brain research}, 1313:202--214, 2010.

\bibitem{Wing}
Lorna Wing.
\newblock The autistic spectrum.
\newblock {\em The Lancet}, 350(9093):1761--1766, 1997.

\end{thebibliography}
\bibliographystyle{plain}

\end{document}